\documentclass[letterpaper,twocolumn]{article} 
\usepackage{times}

\usepackage{graphicx}
\usepackage{latexsym}

\usepackage{amsthm}
\usepackage{amssymb}	
\usepackage{multirow}
\usepackage{pgfplots}
\usepackage{tikz}
\usetikzlibrary{shapes,arrows}
\usetikzlibrary{positioning}

\pgfplotsset{width=7cm,compat=1.8}
\usepackage{pgfplotstable}

\usepackage{helvet}  
\usepackage{courier}
\usepackage[hyphens]{url} 
\usepackage{graphicx}
\usepackage{caption} 
\usepackage{algorithm}
\usepackage{algorithmic}
\usepackage{newfloat}
\usepackage{listings}
\usepackage{geometry}
\floatstyle{ruled}
\newfloat{listing}{tb}{lst}{}
\floatname{listing}{Listing}
\usepackage[utf8]{inputenc}
\usepackage[english]{babel}
\usepackage{enumerate}
\usepackage{ amssymb}
\usepackage{amsmath}
\usepackage{amsthm}
\usepackage{ dsfont }
\usepackage{xspace}
\usepackage{subcaption}
\usepackage{etoolbox}
\def\ssk{\smallskip}
\def\pagebreak{\vfill\eject}
\def\leftdisplay#1$${\leftline{\hskip 36pt$\displaystyle{#1}$}$$}

\newcount\commenton   %edited early 2002
\def\comment#1 \par{\ifnum\commenton=1
    \medskip{\it [#1]}\par\medskip\fi}
\newcount\commentaon       %less important comments
\def\commenta#1 \par{\ifnum\commentaon=1
    \medskip{\it [#1]}\par\medskip\fi}
\newcount\commentbon
\def\commentb#1 \par{\ifnum\commentbon=1
    \medskip{\it [#1]}\par\medskip\fi}

\newcount\commentcon
\def\commentc#1 \par{\ifnum\commentcon=1
    \medskip{#1}\par\medskip\fi}

   \newcount\commentpion %proofs to be included
\def\commentpi#1 \par{\ifnum\commentpion=1
    {{#1}}\par\fi}
\def\commentpistart#1 \par{\ifnum\commentpion=1 %start of proof
    \medskip{\noindent {\em Proof:~}}{{#1}}\par\fi}
\def\commentpiend#1 \par{\ifnum\commentpion=1 %end of proof
    {{#1}}\hfill $\Box$\par\medskip\fi}
\def\commentpistartend#1 \par{\ifnum\commentpion=1 %start of proof
    \medskip{\noindent {\em Proof:~}}{{#1}}\hfill $\Box$\par\medskip\fi}

    \newcount\commentphon %use e.g. for proofs which probably
\def\commentph#1 \par{\ifnum\commentphon=1
    {{#1}}\par\fi}
\def\commentphstart#1 \par{\ifnum\commentphon=1 %start of proof
    \medskip{\noindent {\em Proof:~}}{{#1}}\par\fi}
\def\commentphend#1 \par{\ifnum\commentphon=1 %end of proof
    {{#1}}\hfill $\Box$\par\medskip\fi}
\def\commentphstartend#1 \par{\ifnum\commentphon=1 %start of proof
    \medskip{\noindent {\em Proof:~}}{{#1}}\hfill $\Box$\par\medskip\fi}

\def\summ_#1{\setbox1=\hbox{$\displaystyle\sum_{#1}$}
             \setbox2=\hbox{$\displaystyle\sum$}
             \dimen1=0.25\wd1
             \dimen2=0.25\wd2
             \advance \dimen1 by -\dimen2
             \dimen3 = 0.5\dimen1
             \hskip -\dimen3
             \hbox {$\displaystyle\sum_{#1}$}
             \hskip -\dimen1
             }

\def\Nats{{\hbox{$\mathpalette{}{I\kern-.33em N}$}}}
\def\Reals{{\hbox{$\mathpalette{}{I\kern-.33em R}$}}}

\newcommand{\set}[1]{{\{#1\}}}
\newcommand{\emp}{\emptyset}

\newcommand{\calA}{{\cal A}}

\newcommand{\calG}{{\cal G}}

\newcommand{\calL}{{\cal L}}
\newcommand{\calM}{{\cal M}}

\newcommand{\calO}{{\cal O}}
\newcommand{\calP}{{\cal P}}

\def\Nats{{\hbox{$\mathpalette{}{I\kern-.33em N}$}}}
\def\Reals{{\hbox{$\mathpalette{}{I\kern-.33em R}$}}}
\def\set#1{{\{#1\}}}

\def\st{\, : \,}    % such that---used in definition of sets

\def\emp{\emptyset}

\def\calA{{\cal A}}

\def\calG{{\cal G}}

\def\calL{{\cal L}}
\def\calM{{\cal M}}

\def\calO{{\cal O}}
\def\calP{{\cal P}}

\newtheorem{proposition}{Proposition}
\newtheorem{lemma}{Lemma}
\newtheorem{corollary}{Corollary}

\newtheorem{definition}{Definition}
\newtheorem{theorem}{Theorem}

\def\und{\underline}

\def\phi{\varphi}

\renewcommand{\AA}{A}

\def\succceq{\succcurlyeq}

\def\eqq{$\parallel$}

\newcommand{\NO}{{\rm NO}}

\newcommand{\PO}{{\rm PO}}
\newcommand{\PSO}{{\rm PSO}}
\newcommand{\OO}{{\rm O}}
\newcommand{\SO}{{\rm SO}}

\newcommand{\CSD}{{\rm CSD}}

\newcommand{\MPO}{{\rm MPO}}

\newcommand{\EXT}{{\rm EXT}}

\newcommand{\Opt}{\textrm{Opt}}
\newcommand{\POM}{{\rm POM}}

\newcommand{\RR}{\mathfrak{R}}

\newcommand{\alphabeta}{(\alpha, \beta)}
\newcommand{\gralphabeta}{\alpha \ge \beta}

\newcommand{\Glex}{\calG}
\newcommand{\modelslex}{\models}

\newcommand{\modelsstar}{\models^*}
\newcommand{\calLpq}{\calL_{pqT}}
\newcommand{\calLpqn}{\calL'_{pqT}}

\newcommand{\op}{\Omega}

\newcommand{\shortspecialcell}[2][c]{\begin{tabular}[#1]{@{}p{2cm}@{}}#2\end{tabular}}

\title{
Efficient  Inference and Computation of Optimal Alternatives for Preference Languages Based On Lexicographic Models
\\ \textit{Longer Version of IJCAI-17 Submission}
}

\author{Nic Wilson, Anne-Marie George \\
Insight Centre for Data Analytics,
School of Computer Science and IT \\
University College Cork, Ireland \\
\{nic.wilson, annemarie.george\}@insight-centre.org
}

\begin{document}

\maketitle

\begin{abstract}
We analyse preference inference, through consistency,
for general preference languages based on lexicographic models.
We identify a property, which we call \emph{strong compositionality},
that applies for many natural kinds of preference statement,
and
that allows a greedy algorithm for determining consistency of a set of preference statements.
We
also consider different natural definitions of optimality, and their relations to each other, for general preference languages based on lexicographic models.
Based on our framework, we show that testing consistency, and thus inference, is polynomial for a specific preference language $\calLpqn$,
which allows strict and non-strict statements, comparisons between outcomes and between partial tuples, both ceteris paribus and strong statements, and their combination.
Computing different kinds of optimal sets is also shown to be polynomial; this is backed up by our experimental results.
\end{abstract}

\newcommand{\outc}{\und{V}}
\newcommand{\Tr}{\textit{Tr}}
\newcommand{\RTr}{\textit{RTr}}

\newcommand{\Gammadestrict}{\Gamma^{(\ge)}}
\newcommand{\phidestrict}{\phi^{(\ge)}}
\newcommand{\psidestrict}{\psi^{(\ge)}}

\section{Introduction}
\label{sec:intro}

Preferences are considered in many different fields such as recommender systems and human-computer interaction~\cite{ChenP07,TrabelsiWB13}, databases~\cite{AgrawalW00,KieBling02} and multi-objective decision making and social choice~\cite{ArrowR86,Lang02,SandholmB06}.
They are used to, e.g., find optimal decisions, solutions or items that satisfy the wishes or needs of a user or a group.
In order to reason with preferences, typically assumptions are made on the
form of the relation used to model the user preferences.
In this paper we restrict our considerations to lexicographic models,
and we consider preference inference,
which involves reasoning about the the whole set of preference models that are consistent with the input preferences, as in~\cite{KohliJ07,Wilson14};
this contrasts with
 work that focuses on learning one lexicographic model that fits best with the given preferences~\cite{BoothCLMS10,BrauningH12,DombiIV07,YamanWLdJ10}.

Preference inference aims to overcome gaps of knowledge in the user preferences by analysing the given preferences and deciding whether another preference statement can be deduced.
Since in this paper we allow negations of statements, preference inference
can be reduced to testing consistency of a set of preference statements.
We define a property, strong compositionality, that is satisfied by many natural types of preference statements and
enables a simple greedy algorithm for testing consistency.

Fast computation is essential in many applications;
we give a concrete instance $\calLpqn$ of a preference language that allows polynomial computation.
This can represent a relatively expressive form of preference input.
It can express preferences between a pair of complete assignments to a set of variables, and, more generally, between partial assignments
to just a subset of the variables, allowing \emph{ceteris paribus} assumptions or stronger implications on the values of other variables.
We also allow non-strict statements, two forms of strict preference statements and can represent certain negated statements.

Previous work
on preference inference based on standard lexicographic models
 have considered more restricted preference languages.
Wilson~\cite{Wilson14} considered only non-negated non-strict statements,
which can only express that one assignment is at least as good as another (or equivalent).
Kohli and Jedidi \cite{KohliJ07}
%on the other hand
considers only non-negated strict statements, which can only express that one complete assignment is strictly better than another.
In~\cite{WilsonGOS15}, preference statements are comparisons of complete assignments,
and fixed value orders for the variables are given.

Other work on preference inference includes that based on
weighted sum models~\cite{WilsonRM15,MontazeryW16},
on hierarchical models~\cite{WilsonGOS15},
conditional lexicographic models~\cite{Wilson09};
Pareto orders~\cite{GeorgeW16}, and on general strict total orders, as in e.g., work on conditional preference structures such as \cite{BoutilierBDHP04-CP-NetsATool}.
While lexicographic models prevent tradeoffs between variables, they usually allow  better complexity results for general preference languages.

Given a set of preference inputs and a set of alternatives,
there are several natural notions of optimal set; for instance being undominated with respect to the induced preference relation;
or being possibly optimal, i.e., optimal in at least one model.
We establish relationships between the different notions of optimality.
For preference language $\calLpqn$, we analyse the running times and set relations theoretically and experimentally, illustrating the efficiency of the algorithms.

%\commentphnot
%The paper is organised as follows: Section~\ref{sec:lex-models} gives basic definitions.
%The concepts of  inference, consistency and strongly compositional statements are investigated in Section~\ref{sec:lex-inf-strong-compos}.
%Section~\ref{sec:pref-languages-calLpq} introduces languages $\calLpq$ and $\calLpqn$ and an algorithm to decide consistency in this context.
%In Section~\ref{sec:optimality}, different notions of optimality are compared for general preference languages together with a discussion of computational methods and complexity results for $\calLpqn$.
%The last section concludes. A longer version of this paper including proofs and additional comments can be found under~\cite{WilsonG17longer}.

%\commentph
The paper is organised as follows: Section~\ref{sec:lex-models} gives basic definitions.
The concepts of  inference, consistency and strongly compositional statements are investigated in Section~\ref{sec:lex-inf-strong-compos}.
Section~\ref{sec:pref-languages-calLpq} introduces languages $\calLpq$ and $\calLpqn$ and an algorithm to decide consistency in this context.
In Section~\ref{sec:optimality}, different notions of optimality are compared for general preference languages.
The efficiency of computing these for $\calLpqn$ is analysed theoretically and experimentally in Section~\ref{sec: experiments}.
The last section concludes.

\newcommand{\dom}{D}

\section{Lex Models, Composition and Extension}
\label{sec:lex-models}

In this section, we give some basic definitions; in particular, we define lexicographic models, along with
a natural composition operation and extension relation, that are important in our approach.
A lexicographic model has an associated ordering of (some of the) variables, along with value orderings for each such variable.
This generates a (lexicographic) ordering relation, by first comparing outcomes $\alpha$ and $\beta$ on the first (i.e., most important) variable $Y$; only if $\alpha(Y) = \beta(Y)$ is the next most important variable considered.

Throughout the paper we consider a fixed set $V$ of $n$ variables,
where for each $X \in V$, $\dom(X)$
is the set of possible values of $X$.
For subset of variables $A \subseteq V$ let $\und{A} = \prod_{X \in A} \dom(X)$ be the
set of possible assignments to $A$.
For $X\in V$, we abbreviate $\und{\set{X}}$ to $\und{X}$
(which is essentially $\dom(X)$).
An \emph{outcome} is an element of $\und{V}$,
i.e., an assignment to all the variables.
For outcome $\alpha$ and subset $U$ of  $V$,
we define $\alpha(U)$ to be the tuple in $\und{U}$ generated by projecting (i.e., restricting) $\alpha$ to $U$.
These definitions describe a decision scenario over a set of alternatives given by different features.
Consider for example the choice between different flight connections.
A set of variables to describe the flight connections could be $V = \{\mathit{airline}, \mathit{time}, \mathit{class}\}$.
The variable domains are then given by $\dom(\mathit{airline}) = \und{\mathit{airline}} = \{\mathit{KLM}, \mathit{LAN}\}$, $\dom(\mathit{time}) = \{\mathit{day}, \mathit{night}\}$ and $\dom(\mathit{class}) = \{\mathit{economy}, \mathit{business}\}$.
Then projection of the outcome $\alpha = (\mathit{KLM}, \mathit{day}, \mathit{economy})$ to the variables $U = \{\mathit{airline}, \mathit{class}\}$ is given by $\alpha(U) = (\mathit{KLM}, \mathit{economy})$.

Define $\Glex$ to be the set of lexicographic models (over the set of variables $V$);
a \emph{lexicographic model} or \emph{lex model}, $\pi$ (over $V$),  is defined
to be a (possibly empty) sequence
of the form
$(Y_1, \ge_{Y_1}), \ldots, (Y_k, \ge_{Y_k})$,
 where  $Y_i$ ($i=1, \ldots, k$) are different variables in $V$,
and each $\ge_{Y_i}$ is a total order on $\und{Y_i}$.
The associated relation $\succceq_\pi$ $\subseteq \outc\times\outc$
is defined as follows:
for outcomes $\alpha$ and $\beta$,
$\alpha \succceq_\pi \beta$ if and only if
either (i) for all $i=1, \ldots, k$,  $\alpha(Y_i) = \beta(Y_i)$;
or
(ii) there exists $i \in \set{1, \ldots, k}$ such that
for all $j < i$, $\alpha(Y_j) = \beta(Y_j)$ and
$\alpha(Y_i) >_{Y_i} \beta(Y_i)$ (i.e., $\alpha(Y_i) \ge_{Y_i} \beta(Y_i)$ and $\alpha(Y_i) \not= \beta(Y_i)$).
Thus $\succceq_\pi$ is a total pre-order on $\und{V}$, which is a total order if
%and only if
$k=n = |V|$.

The corresponding strict relation $\succ_\pi$ is given by
$\alpha \succ_\pi \beta$ if and only if there exists $i \in \set{1, \ldots, k}$ such that
$\alpha(Y_i) >_{Y_i} \beta(Y_i)$
and for all $j < i$, $\alpha(Y_j) = \beta(Y_j)$.
The corresponding equivalence relation $\equiv_\pi$
is given by
$\alpha \equiv_\pi \beta$ if and only if
for all $i=1, \ldots, k$,  $\alpha(Y_i) = \beta(Y_i)$.
Thus, $\alpha \equiv_\pi \beta$ if and only if
$\alpha(V_\pi) = \beta(V_\pi)$,
where $V_\pi=\set{Y_1, \ldots, Y_k}$ is set of the variables involved in $\pi$.
Consider lexicographic models $\pi = (\mathit{airline}, \mathit{KLM} > \mathit{LAN}), (\mathit{time}, \mathit{day} > \mathit{night})$
and $\pi' = (\mathit{class}, \mathit{economy} > \mathit{business}), (\mathit{time}, \mathit{night} > \mathit{day})$
for the previous example of flight connections.
In $\pi$ the choice of airline decides which connection is preferred;
only if these are the same, is the flight time considered.
$\pi'$ first compares the classes;
only if these are the same, are the flight times considered.
For outcomes $\alpha = (\mathit{KLM}, \mathit{day}, \mathit{economy})$, $\beta = (\mathit{KLM}, \mathit{night}, \mathit{business})$
and $\gamma = (\mathit{LAN}, \mathit{day}, \mathit{economy})$ we have $\alpha \succ_\pi \beta \succ_\pi \gamma$
and $\alpha \equiv_{\pi'} \gamma \succ_{\pi'} \beta$.

\ssk\noindent
\textbf{Composition of lexicographic models:}
We define an important composition operation on lex models,
which can be shown to be associative.
Let
$\pi = (Y_1, \ge_{Y_1}), \ldots, (Y_k, \ge_{Y_k})$,
and
$\pi' = (Z_1, \ge_{Z_1}), \ldots, (Z_l, \ge_{Z_l})$
be two lexicographic models.
Let $\pi''$ be the sequence $\pi'$ but where pairs $(Z_i, \ge_{Z_i})$
are omitted if $Z_i \in V_\pi$.
Define lex model $\pi\circ\pi'$ to be
$\pi$ followed by $\pi''$.
Note that $V_{\pi\circ\pi'} = V_\pi \cup V_{\pi'}$.

%\commentph
\begin{lemma}
\label{le:circ-associative}
$\circ$ is an associative operation on $\Glex$.
\end{lemma}

\newcommand{\downaminus}{\downarrow -}

\begin{proof}
%
%\commentphstart
For $\pi,\pi'\in\Glex$,
we say that $\pi$ and $\pi'$ are disjoint if
$V_\pi \cap V_{\pi'} = \emp$.
For
set of variables $U\subseteq V$,
define $\pi^{\downaminus U}$ to be
$\pi$ with all pairs of the form $(X, \ge_X)$ omitted,
when $X\in U \cap V_\pi$.
For any $\pi,\pi'\in\Glex$,
we have that $\pi\circ\pi' = \pi\circ(\pi')^{\downaminus V_\pi}$,
by definition.

%\commentph
Consider any $\pi_1, \pi_2, \pi_3 \in\Glex$.
For $j=1, 2, 3$, let $V_j$ equal $V_{\pi_j}$.
Then, $(\pi_1\circ\pi_2)\circ\pi_3 =
(\pi_1\circ\pi'_2)\circ\pi'_3$,
where
$\pi'_2 = (\pi_2)^{\downaminus V_1}$ and
$\pi'_3 = (\pi_3)^{\downaminus V_1 \cup V_2}$.

%\commentphend
Associativity clearly holds for pairwise disjoint lexicographic models,
so
$(\pi_1\circ\pi'_2)\circ\pi'_3 =
\pi_1 \circ(\pi'_2 \circ \pi'_3)$.
Now, $\pi_1 \circ(\pi_2 \circ \pi_3)$
equals $\pi_1 \circ(\pi_2 \circ \pi_3)^{\downaminus V_1}$,
so, to prove associativity of $\circ$ it is sufficient
to prove that
$(\pi_2 \circ \pi_3)^{\downaminus V_1}$ equals
$\pi'_2 \circ \pi'_3$.
We have that
$(\pi_2 \circ \pi_3)^{\downaminus V_1}$ equals
$(\pi_2 \circ (\pi_3)^{\downaminus V_2})^{\downaminus V_1}$,
which equals
$\pi_2^{\downaminus V_1} \circ \pi_3^{\downaminus V_1 \cup V_2}$,
i.e., $\pi'_2 \circ \pi'_3$, completing the proof.
\end{proof}

\begin{definition}
For lex models $\pi$ and $\pi'$ in $\Glex$
we say that $\pi'$ \emph{extends} $\pi$ if $\pi' \not=\pi$
and the sequence $\pi'$ begins with $\pi$.
We then write $\pi' \sqsupset \pi$,
and write $\pi'\sqsupseteq\pi$ to mean
that $\pi'$ extends or equals $\pi$.
This holds if and only if
%$\pi' \not=\pi$ and
there exists a lex model $\pi''$ such that
$\pi' = \pi\circ\pi''$.
\end{definition}

%\commentph
\begin{lemma}
For $\pi,\pi'\in\Glex$, we have
$\pi'\sqsupseteq\pi$ if and only if $\pi\circ\pi' = \pi'$.
\end{lemma}

\begin{proof}
%\commentphstart
Suppose that $\pi'$ extends or equals $\pi$, and
write $\pi'$ as $\pi\circ\pi''$ where $V_{\pi''} \cap V_\pi = \emp$.
Then, $\pi\circ\pi' = \pi\circ(\pi\circ\pi'')$ which equals
$\pi\circ\pi'' = \pi'$.
%
%\commentphend
Conversely, suppose that $\pi\circ\pi' = \pi'$.
Let $k = |V_{\pi}|$ and let
$\pi''$ be the lexicographic model consisting of the first $k$ pairs
of $\pi'$.
This equals the first $k$ pairs of $\pi\circ\pi'$, i.e., $\pi$,
and thus, $\pi'$ extends or equals $\pi$.
\end{proof}

The composition of the two previously defined lex models is $\pi \circ \pi' =
(\mathit{airline}, \mathit{KLM} > \mathit{LAN}),
(\mathit{time}, \mathit{day} > \mathit{night}),
(\mathit{class}, \mathit{economy} > \mathit{business})$ and $\pi \circ \pi' \sqsupseteq \pi$.

%\commentph
We have the following basic monotonicity property.
\begin{lemma}
\label{le:extends-pi-basic}
For $\pi,\pi'\in\Glex$,
suppose that $\pi'$ extends $\pi$.
If $\alpha \succceq_{\pi'} \beta$ then $\alpha \succceq_{\pi} \beta$.
Also, if $\alpha \succ_{\pi} \beta$ then $\alpha \succ_{\pi'} \beta$.
\end{lemma}

\begin{proof}
%\commentphstartend
If $\alpha \succ_{\pi} \beta$ then
$\alpha$ and $\beta$ differ on some variable in $\pi$, and
$\alpha$ is better than $\beta$
on the first such variable;
 thus the same holds for $\pi'$, since $\pi'$ extends $\pi$,
  and so $\alpha \succ_{\pi'} \beta$ holds.
Similarly,
if $\beta \succ_{\pi} \alpha$ then $\beta \succ_{\pi'} \alpha$,
which, using the fact that $\succ_\pi$ and $\succ_{\pi'}$ are both weak orders,
is equivalent to:
if $\alpha \not\succceq_{\pi} \beta$ then $\alpha \not\succceq_{\pi'} \beta$,
which thus implies:
if $\alpha \succceq_{\pi'} \beta$ then $\alpha \succceq_{\pi} \beta$.
\end{proof}

%\commentph
The result below is a key property of composition,
which will be important later.
%
%\commentph
\begin{lemma}
\label{le:composition-basic}
Let $\pi,\pi'\in\Glex$.
If $\alpha \succceq_\pi \beta$ and $\alpha \succceq_{\pi'} \beta$
then $\alpha \succceq_{\pi\circ\pi'} \beta$.
If $\alpha \succceq_\pi \beta$ and $\alpha \succ_{\pi'} \beta$
then $\alpha \succ_{\pi\circ\pi'} \beta$
and $\alpha \succ_{\pi'\circ\pi} \beta$.
Suppose now that $\alpha \equiv_\pi \beta$. Then
$\alpha \succceq_{\pi'} \beta$ $\iff$ $\alpha \succceq_{\pi\circ\pi'} \beta$;
and
$\alpha \succ_{\pi'} \beta$ $\iff$ $\alpha \succ_{\pi\circ\pi'} \beta$.
\end{lemma}

\begin{proof}
%
%\commentphstart
Since $\pi'\circ\pi$ extends $\pi'$,
we have, by Lemma~\ref{le:extends-pi-basic},
that (i)  $\alpha \succ_{\pi'} \beta$
implies  $\alpha \succ_{\pi'\circ\pi} \beta$.
Similarly, (ii) if $\alpha \succ_{\pi} \beta$
then $\alpha \succ_{\pi\circ\pi'} \beta$,
and thus also $\alpha \succceq_{\pi\circ\pi'} \beta$.
We also have
(iii) if $\alpha \equiv_{\pi} \beta$ and $\alpha \succ_{\pi'} \beta$
then $\alpha \succ_{\pi\circ\pi'} \beta$.
This is because
 $\alpha$ and $\beta$ differ on some variable in $\pi\circ\pi'$,
and the first such variable pair $(X, \ge)$ appears in
$\pi'$; and $X$ is the earliest variable in $\pi'$ that $\alpha$
and $\beta$ differ on, so $\alpha \succ_{\pi'} \beta$ implies that
 $\alpha(X) > \beta(X)$ and thus, $\alpha \succ_{\pi\circ\pi'} \beta$.

%\commentph
Suppose that $\alpha \succceq_\pi \beta$ and $\alpha \succ_{\pi'} \beta$.
Then, by (i), $\alpha \succ_{\pi'\circ\pi} \beta$.
If  $\alpha \succ_{\pi} \beta$ then $\alpha \succ_{\pi\circ\pi'} \beta$, by (ii).
So, we only need to deal with the case where $\alpha \equiv_{\pi} \beta$.
Then
$\alpha \succ_{\pi\circ\pi'} \beta$ follows by (iii).

%\commentph
Suppose that $\alpha \succceq_\pi \beta$ and $\alpha \succceq_{\pi'} \beta$.
If $\alpha \succ_{\pi} \beta$ or $\alpha \succ_{\pi'} \beta$ then, by the above arguments,
we have $\alpha \succ_{\pi\circ\pi'} \beta$,
and thus also $\alpha \succceq_{\pi\circ\pi'} \beta$.
So, it just remains to deal with the case where
$\alpha \equiv_\pi \beta$ and $\alpha \equiv_{\pi'} \beta$.
But then $\alpha(V_\pi) = \beta(V_\pi)$
and $\alpha(V_{\pi'}) = \beta(V_{\pi'})$
and so, because $V_{\pi\circ\pi'} = V_\pi \cup V_{\pi'}$,
we have $\alpha(V_{\pi\circ\pi'}) = \beta(V_{\pi\circ\pi'})$,
proving $\alpha \equiv_{\pi\circ\pi'} \beta$
and hence $\alpha \succceq_{\pi\circ\pi'} \beta$.

%\commentphend
%LAST PART done later:
Assume now that $\alpha \equiv_\pi \beta$,
so that $\alpha \succceq_\pi \beta$ and $\beta \succceq_\pi \alpha$.
By the earlier parts we have
 $\alpha \succceq_{\pi'} \beta$ $\Rightarrow$ $\alpha \succceq_{\pi\circ\pi'} \beta$;
 and $\alpha \succ_{\pi'} \beta$ $\Rightarrow$ $\alpha \succ_{\pi\circ\pi'} \beta$.
Since $\beta \succceq_\pi \alpha$,
we also have if $\alpha \not\succceq_{\pi'} \beta$, i.e.,
$\beta \succ_{\pi'} \alpha$, then $\beta \succ_{\pi\circ\pi'} \alpha$,
i.e., $\alpha \not\succceq_{\pi\circ\pi'} \beta$.
Similarly we have $\alpha \not\succ_{\pi'} \beta$ implying
$\alpha \not\succ_{\pi\circ\pi'} \beta$.
\end{proof}

\section{Lexicographic Inference and Strongly Compositional Preference Statements}
\label{sec:lex-inf-strong-compos}

In this section we define \emph{lexicographic inference},
which can also be expressed in terms of (lexicographic) consistency.
We define a notion of \emph{strong compositionality},
and show that
this property
 enables a greedy algorithm for checking consistency, and hence for preference inference.
More specifically, a preference statement $\phi$ is strongly compositional if
 the composition of two lex models satisfies $\phi$ whenever the second and some extension of the first satisfy $\phi$.

Throughout this section, we assume some language $\calL$, and satisfaction relation $\models\, \subseteq  \Glex\times\calL$.
Here, the language $\calL$ is a set of preference statements and the relation $\models$ describes when a lex model satisfies a statement.
For example, we can consider the simple language $\calL^{O} = \{\alpha \geq \beta \mid \alpha,\beta \in \und{V}\}$ with the satisfaction relation $\pi \models \alpha \geq \beta \Leftrightarrow \alpha \succceq_\pi \beta$ as described in the previous section.
We extend the satisfaction relation to subsets $\Gamma$ of $\calL$ in the usual way:
for $\Gamma\subseteq\calL$,
define $\pi \models \Gamma$ $\iff$
$\pi\models \phi$ for each $\phi\in\Gamma$.
If $\pi \models\Gamma$ then we say that \emph{$\pi$ satisfies $\Gamma$},
or $\pi$ is a \emph{[lexicographic] model} of $\Gamma$
(and similarly, for a single preference statement $\phi$).
We say that $\Gamma$ is \emph{consistent},
if there exists
some $\pi\in\Glex$ satisfying $\Gamma$;
otherwise, $\Gamma$ is \emph{inconsistent}.
We will sometimes use negations of preference statements:
when necessary we can always extend the language $\calL$
and the relation $\models$,
by defining $\pi \models \neg\phi$ $\iff$ $\pi\not\models\phi$, when $\phi\in\calL$.

\ssk
A key problem is to determine if a set $\Gamma\subseteq\calL$ of preference statements is consistent.
We will derive methods for determining this in Section~\ref{subsec:checking-cons}.

\subsubsection*{Some Basic Preference Statements}

For outcomes $\alpha, \beta\in\outc$, we can define preference statements $\gralphabeta$ and $\alpha>\beta$,
 where we define
$\pi \modelslex \gralphabeta$ $\iff$ $\alpha \succceq_\pi \beta$,
and
$\pi \modelslex \alpha>\beta$ $\iff$ $\alpha \succ_\pi \beta$.
We also define preference statement $\alpha\equiv\beta$, with
$\pi\modelslex \alpha\equiv\beta$ $\iff$ $\alpha \equiv_\pi \beta$,
which is if and only if
$\pi\modelslex\gralphabeta$
and $\pi\modelslex\beta\ge\alpha$.

We have that $\pi\models \gralphabeta$
$\iff$ $\pi\not\models\beta>\alpha$,
i.e., $\pi\models \neg(\beta>\alpha)$.
So $\gralphabeta$ and $\neg(\beta>\alpha)$ are equivalent preference statements,
in that they are satisfied by exactly the same set of lex models.
Similarly,
$\alpha>\beta$ and $\neg(\beta\ge\alpha)$ are equivalent preference statements.

Consider the example of lex models $\pi,\pi'$ and outcomes $\alpha,\beta,\gamma$ as before.
Then $\pi \models \alpha \geq \beta$ and thus $\pi\not\models\beta>\alpha$.
Also, $\pi' \modelslex \alpha\equiv\gamma$, i.e., $\pi' \modelslex \alpha \geq \gamma$ and $\pi' \modelslex\gamma\ge\alpha$.
%\commentph
%\begin{example} %[continued]
%\label{ex: continued-basic-statements}
%Consider the flight connections from Example~\ref{ex: composition-and-extension}.
%Let $\alpha = (\mathit{day}, \mathit{KLM}, \mathit{economy})$ and $\beta = (\mathit{day}, \mathit{LAN}, \mathit{business})$.
%Then the model $\pi = (\mathit{time}, \mathit{night} > \mathit{day})$ satisfies $\pi\modelslex \alpha\equiv\beta$.
%The extension
%$\pi' = (\mathit{time}, \mathit{night} > \mathit{day}),
%(\mathit{class}, \mathit{economy} > \mathit{business})
%$
%of $\pi$ satisfies $\pi' \modelslex \gralphabeta$,
%and in fact $\pi' \modelslex \alpha>\beta$.
%Thus, $\pi\not\models\beta>\alpha$.
%\end{example}

\subsection{Lexicographic Inference}
\label{subsec:lex-inference}

We define the lexicographic inference relation
$\Gamma \modelslex \phi$ $\iff$
$\pi\models\phi$ for all $\pi \in \Glex$ such that $\pi\models\Gamma$.
For instance, we have, for $\Gamma\subseteq\calL$ and $\alpha, \beta\in\outc$,
$\Gamma\modelslex \gralphabeta$ $\iff$
$\alpha \succceq_\pi \beta$ holds for all $\pi \in \Glex$ such that $\pi\models\Gamma$.
Similarly,
$\Gamma\modelslex \alpha>\beta$ $\iff$
$\alpha \succ_\pi \beta$ holds for all $\pi \in \Glex$ such that $\pi\models\Gamma$.

Lexicographic inference can be reduced to checking consistency,
since we have
$\Gamma \modelslex \phi$ $\iff$
$\Gamma\cup\set{\neg\phi}$ is inconsistent;
and
$\Gamma \modelslex \neg\phi$ $\iff$
$\Gamma\cup\set{\phi}$ is inconsistent.
Thus,
$\Gamma\modelslex \gralphabeta$ $\iff$
$\Gamma\cup\set{\beta>\alpha}$ is inconsistent;
and $\Gamma\modelslex \alpha>\beta$ $\iff$
$\Gamma\cup\set{\beta\ge\alpha}$ is inconsistent.

Consider flight connections as before.
Let $\Gamma$ be the set of statements $\{\alpha > \beta, \beta \geq \gamma\}$ with $\alpha, \beta, \gamma$ defined as in the previous section.
Then the lex model $\pi$ satisfies $\Gamma$, $\pi \models \Gamma$, and thus $\Gamma$ is consistent.
In fact, the only preference models that satisfy $\Gamma$ are:
$\pi = (\mathit{airline}, \mathit{KLM} > \mathit{LAN}), (\mathit{time}, \mathit{day} > \mathit{night})$, $\pi'' = (\mathit{airline}, \mathit{KLM} > \mathit{LAN}), (\mathit{class}, \mathit{economy} > \mathit{business})$,
and any extension of either of these two models.
Thus for $\delta = (\mathit{LAN}, \mathit{night}, \mathit{business})$,
$\Gamma \models \gamma > \delta$.

\subsubsection*{The Induced Relation $\modelsstar$
and Maximal Models}
\label{subsec:modelsstar}

From the relation $\models$ we also define the derived relation $\modelsstar$
as follows, where $\pi\in\Glex$ and $\phi\in\calL$:
$\pi \modelsstar \phi$ if and only if
there exists $\pi' \in\Glex$ either extending or equalling $\pi$
such that $\pi'\models\phi$.
Thus, $\pi \modelsstar \phi$ holds either if
$\pi$ satisfies $\phi$ or some extension of $\pi$ satisfies $\phi$.
We extend the relation to sets of statements in the usual way:
for $\Gamma\subseteq\calL$,
define $\pi \modelsstar \Gamma$ if and only if
$\pi \modelsstar \phi$ holds for every $\phi\in\Gamma$.
The following lemma follows easily from the definitions.

\begin{lemma}
\label{le:modelsstar-basic}
Let $\pi, \pi'\in\Glex$
and $\Gamma\subseteq\calL$.
\begin{itemize}
  \item[(i)]
  $\pi\models\Gamma$ $\Rightarrow$ $\pi\modelsstar\Gamma$.
  \item[(ii)] Suppose that $\pi'$ extends $\pi$.
  Then
  $\pi'\modelsstar\Gamma$ $\Rightarrow$ $\pi\modelsstar\Gamma$.
\end{itemize}
\end{lemma}

\begin{proof}
%\commentphstartend
(i) follows immediately from the definition of $\modelsstar$.
Regarding (ii), $\pi'\modelsstar\phi$ implies
that there exists $\pi''$ with $\pi''\sqsupseteq\pi'$
and $\pi''\models\phi$. Then, $\pi''\sqsupseteq\pi$,
and thus, $\pi\modelsstar\phi$.
The result then follows.
\end{proof}

\begin{definition}
We say that $\pi\in\Glex$ is a \emph{maximal model of $\Gamma$}
if $\pi \models \Gamma$ and for all lex models $\pi'$
extending $\pi$ we have $\pi' \not\models\Gamma$.
\end{definition}

\noindent
Clearly, for $\pi\in\Glex$ such that $\pi\models\Gamma$,
either $\pi$ is a maximal model of $\Gamma$
or there exists a maximal model of $\Gamma$ that extends $\pi$.
Thus, if $\Gamma$ is consistent there exists a maximal model of $\Gamma$.
Analogously, we define a maximal $\modelsstar$-model of $\Gamma$ to
be an element $\pi\in\Glex$
such that (i) $\pi\modelsstar\Gamma$
and (ii) there does not exist $\pi'$ extending $\pi$
with $\pi'\modelsstar\Gamma$.

Consider preference statements $\Gamma$ relating to
flight connections $\alpha, \beta, \gamma$ as before.
The lex model $(\mathit{airline}, \mathit{KLM} > \mathit{LAN})$ satisfies
$(\mathit{airline}, \mathit{KLM} > \mathit{LAN}) \modelsstar \Gamma$.
The model $(\mathit{airline}, \mathit{KLM} > \mathit{LAN}),(\mathit{time}, \mathit{day} > \mathit{night}),(\mathit{class}, \mathit{business} > \mathit{economy})$ is a maximal model of $\Gamma$ (and a maximal $\modelsstar$-model of $\Gamma$).
%\commentph
%\begin{example} %[continued]
%\label{ex:continued-basic-statements-further}
%Consider flight connections as in Example~\ref{ex: composition-and-extension}.
%Assume a user states the following preferences:
%\begin{tabbing}
%	$\Gamma = \{$ 	\= $(\mathit{day}, \mathit{KLM}, \mathit{economy})\;\;$ \= $\geq (\mathit{night}, \mathit{LAN}, \mathit{business})$, \\
%				\> $(\mathit{night}, \mathit{KLM}, \mathit{business})$ \> $> (\mathit{day}, \mathit{KLM}, \mathit{economy})$$\}$.
%\end{tabbing}
%Then, $(\mathit{airline}, \mathit{KLM} > \mathit{LAN}) \modelsstar\Gamma$.
%In fact, only extensions $\pi \sqsupset (\mathit{airline}, \mathit{KLM} > \mathit{LAN})$ satisfy $\pi\models\Gamma$.
%For example, the model $(\mathit{airline}, \mathit{KLM} > \mathit{LAN}),(\mathit{class}, \mathit{business} > \mathit{economy}),(\mathit{time}, \mathit{day} > \mathit{night})$ is a maximal model of $\Gamma$.
%Also, $\Gamma \models (\mathit{night}, \mathit{LAN}, \mathit{business}) > (\mathit{day}, \mathit{LAN}, \mathit{economy})$.
%\end{example}

\subsection{Compositional and Strongly Compositional Preference Statements}
\label{subsec:compos}

We formulate a pair of properties of preference statements that
have strong implications regarding lexicographic inference and optimality.

\begin{definition}[(Strongly) Compositional]
Let $\phi\in\calL$.
We say that $\phi$ is \emph{compositional} if
for all $\pi,\pi'\in\Glex$,

\ssk
$\pi\models\phi$ and $\pi'\models\phi$ implies $\pi\circ\pi'\models\phi$.

\ssk
\noindent
We say that
$\phi$ is \emph{strongly compositional} if
for all $\pi,\pi'\in\Glex$,

\ssk
$\pi\modelsstar\phi$ and $\pi'\models\phi$ implies $\pi\circ\pi'\models\phi$.

\ssk
\noindent
For $\Gamma\subseteq\calL$,
we define $\Gamma$ to be \emph{compositional} if
every element of $\Gamma$ is compositional.
Similarly, we say that $\Gamma$ is \emph{strongly compositional} if
every element of $\Gamma$ is strongly compositional.
\end{definition}

%\commentph
\noindent
Note that if $\phi\in\calL$ is inconsistent
(i.e., there exists no $\pi\in\Glex$ with $\pi\models\phi$)
then $\phi$ is trivially (strongly) compositional.

\subsubsection*{Instances of Strongly Compositional Statements}

Although being strongly compositional might appear to be quite a restrictive assumption,
it turns out that it is satisfied by many natural preference statements, as illustrated by the next proposition.
First, we give a lemma which, roughly speaking, states that the property of being [strongly] compositional is closed under conjunction.

\begin{lemma}
\label{le:compos-conjunction}
Let $\Gamma\subseteq\calL$   and let $\psi\in\calL$.
Suppose that $\psi$ is such that for all $\pi\in\Glex$,
$\pi\models\psi$ $\iff$ $\pi\models\Gamma$.
If $\Gamma$ is compositional then $\psi$ is compositional.
If $\Gamma$ is strongly compositional then $\psi$ is strongly compositional
and, for all $\pi\in\Glex$,
[$\pi\modelsstar\psi$ $\iff$ $\pi\modelsstar\Gamma$ and $\Gamma$ is consistent].
\end{lemma}

\begin{proof}
%\commentphstart
Suppose that $\Gamma$ is compositional.
Consider any $\pi,\pi'\in\Glex$ with $\pi,\pi'\models\psi$.
Then, $\pi,\pi'\models\Gamma$, which, since $\Gamma$ is compositional,
implies $\pi\circ\pi'\models\Gamma$, and thus,
$\pi\circ\pi'\models\psi$, showing that $\psi$ is compositional.

%\commentph
Now assume, for the remainder of the proof of this result, that $\Gamma$ is strongly compositional.
First suppose also that $\pi\modelsstar\psi$.
Then there exists $\pi''$ such that $\pi''\sqsupseteq\pi$
and $\pi''\models\psi$,
which implies that $\psi$ is consistent, and so $\Gamma$ is consistent.
Also, for all $\phi\in\Gamma$ we have $\pi''\models\phi$ and hence, $\pi\modelsstar\phi$. We have shown that $\pi\modelsstar\Gamma$.

%\commentph
For the converse, we assume that $\pi\modelsstar\Gamma$ and $\Gamma$ is consistent.
Consistency of $\Gamma$ implies that there exists $\pi'\in\Glex$
such that $\pi'\models\Gamma$.
Since $\Gamma$ is strongly compositional,
$\pi\circ\pi'\models\Gamma$,
and so, $\pi\circ\pi'\models\psi$
 which implies that $\pi\modelsstar\psi$,
 because $\pi\circ\pi'\sqsupseteq\pi$.

%\commentphend
We will finally show that $\psi$ is strongly compositional.
Consider any $\pi,\pi'\in\Glex$ with $\pi\modelsstar\psi$ and $\pi'\models\psi$.
We have, by the earlier part, that $\pi\modelsstar\Gamma$,
and also $\pi'\models\Gamma$.
$\Gamma$ being strongly compositional implies $\pi\circ\pi'\models\Gamma$
and thus, $\pi\circ\pi'\models\psi$, proving that $\psi$ is strongly compositional.
\end{proof}

\begin{proposition}
\label{pr:strong-compos-alpha-beta-calR}
For any outcomes, $\alpha,\beta\in\und{V}$,
statements $\alpha\ge\beta$ and $\alpha>\beta$
are strongly compositional,
and,
for $\pi\in\Glex$,
$\pi\modelsstar\alpha\ge\beta$
$\iff$ $\pi\models\alpha\ge\beta$.
Also, if $\alpha\not=\beta$,
$\pi\modelsstar\alpha>\beta$ $\iff$ $\pi\models\alpha\ge\beta$.

In addition, given $\RR\subseteq\und{V}\times\und{V}$,
let $\phi^\RR$ be some statement satisfying:
$\pi\models\phi^\RR$ if and only if
$\succceq_\pi \supseteq\RR$.
Then
$\phi^\RR$ is strongly compositional,
and, for any $\pi\in\Glex$, we have $\pi\modelsstar\phi^\RR$ $\iff$
$\pi\models\phi^\RR$
\end{proposition}

%\commentphnot
\noindent
Here, the second half of Proposition~\ref{pr:strong-compos-alpha-beta-calR} follows from the first using Lemma~\ref{le:compos-conjunction}.
\noindent
Further important examples of strongly compositional statements will be discussed in Section~\ref{sec:pref-languages-calLpq}.

\subsubsection*{A Characterisation of Consistency
for Strongly Compositional $\Gamma$}

The definitions immediately imply the following.

\begin{lemma}
\label{le:strong-compos-basic-Gamma}
Let $\Gamma\subseteq\calL$, and let $\pi,\pi'\in\Glex$.
\begin{itemize}
  \item If $\Gamma$ is strongly compositional then it is compositional.
  \item If $\Gamma$ is compositional then $\pi\models\Gamma$ and $\pi'\models\Gamma$ imply $\pi\circ\pi'\models\Gamma$.
  \item If $\Gamma$ is strongly compositional then
   $\pi\modelsstar\Gamma$ and $\pi'\models\Gamma$ imply $\pi\circ\pi'\models\Gamma$.
\end{itemize}
\end{lemma}

\noindent
The last point implies that, for strongly compositional and consistent $\Gamma$,
if $\pi\modelsstar\Gamma$ then there exists a model of
$\Gamma$ either equalling or extending $\pi$.
In fact we have:

\begin{lemma}
\label{le:Gamma-strong-compos-equivalence}
Suppose that $\Gamma$ is strongly compositional,
and let $\pi$ be an element of $\Glex$.
Then [there exists $\pi'\in\Glex$ with $\pi'\sqsupseteq\pi$ and $\pi'\models\Gamma$]
if and only if [$\Gamma$ is consistent and $\pi\modelsstar\Gamma$].
\end{lemma}

\begin{proof}
%\commentphstart
$\Rightarrow$: First assume that there exists $\pi'\in\Glex$ with $\pi'\sqsupseteq\pi$ and $\pi'\models\Gamma$. Clearly, $\Gamma$ is consistent.
Consider any $\phi\in\Gamma$. We have $\pi'\models\phi$,
which implies $\pi\modelsstar\phi$. Therefore, $\pi\modelsstar\Gamma$.

%\commentphend
$\Leftarrow$:
Assume that $\Gamma$ is consistent and $\pi\modelsstar\Gamma$.
Then there exists $\pi'\in\Glex$ with $\pi'\models\Gamma$.
Since $\Gamma$ is strongly compositional,
$\pi\circ\pi'\models\Gamma$, by
Lemma~\ref{le:strong-compos-basic-Gamma},
%which implies that $\Gamma$ is consistent and $\pi\modelsstar\Gamma$,
and we have $\pi\circ\pi'\sqsupseteq\pi$.
\end{proof}

\noindent
Suppose that
$\pi'$ satisfies strongly compositional $\Gamma$
 and that
$\pi$ is a maximal $\modelsstar$-model of $\Gamma$.
 Since we have $\pi\circ\pi'\modelsstar\Gamma$
(because we have $\pi\circ\pi'\models\Gamma$)
and $\pi\circ\pi'\sqsupseteq\pi$,
then the maximality of $\pi$ implies that
$\pi = \pi\circ\pi'$, and thus, $\pi\models\Gamma$.
We also have that $V_\pi \supseteq V_{\pi'}$.
This implies the following theorem and corollary:

\begin{theorem}
\label{th:max-modelsstar-consistent}
If $\Gamma$ is consistent and strongly compositional
then every maximal $\modelsstar$-model of $\Gamma$
satisfies $\Gamma$.
Also, if $\pi$ and $\pi'$ are two maximal $\modelsstar$-models of $\Gamma$
then $V_\pi = V_{\pi'}$.
\end{theorem}

\begin{corollary}
\label{cor:max-modelsstar-consistent}
Let $\pi$ be any maximal $\modelsstar$-model of strongly compositional $\Gamma$.
Then $\Gamma$ is consistent if and only if $\pi\models\Gamma$.
\end{corollary}

\subsection{Checking Consistency of a Set of Strongly Compositional Preferences}
\label{subsec:checking-cons}

Corollary \ref{cor:max-modelsstar-consistent}
shows
that we can test consistency of strongly compositional $\Gamma$,
by finding any  maximal $\modelsstar$-model $\pi$ of it,
and checking if $\pi$ satisfies $\Gamma$.
In this section we show how a  maximal $\modelsstar$-model of $\Gamma$
can be constructed iteratively, starting from the empty model $\set{}$.
We first check if $\set{} \modelsstar \Gamma$.
%**This
The following lemma shows that this
is equivalent to
$\Gamma$ being
$\modelsstar$-consistent,
 which is a very weak property,
i.e., it just requires that each element of $\Gamma$ is (individually) consistent.

\begin{lemma}
\label{le:modelsstar-consistency}
Let $\Gamma\subseteq\calL$.
$\Gamma$ is $\modelsstar$-consistent (i.e., there exists some $\pi\in\Glex$
with $\pi\modelsstar\Gamma$)
if and only if, for each $\phi\in\Gamma$, $\phi$ is consistent.
This also holds if and only if
$\set{} \modelsstar \Gamma$.
\end{lemma}

\begin{proof}
%\commentphstart
$\Rightarrow$:
Suppose that $\Gamma$ is $\modelsstar$-consistent,
and that $\pi\modelsstar\Gamma$.
Then, for any $\phi\in\Gamma$ there exists some
$\pi'\in\Glex$ with $\pi'\sqsupseteq\pi$
and $\pi'\models\phi$, which
implies that $\phi$ is consistent.

%\commentphend
$\Leftarrow$:
Consider any $\phi\in\Gamma$.
Since $\phi$ is consistent, there exists some $\pi\in\Glex$
with $\pi\models\phi$, which entails
that $\set{} \modelsstar \phi$,
where  $\set{}$ is the empty lexicographic order.
Then $\set{}\modelsstar\Gamma$, so $\Gamma$ is $\modelsstar$-consistent.
\end{proof}

\noindent
We say that $\pi'$ \emph{minimally extends} $\pi$ if $\pi'$ extends $\pi$
and $|V_{\pi'}| = |V_\pi|+1$, i.e., $\pi'$ involves one more variable than $\pi$.
The next lemma implies that in building a maximal $\modelsstar$-model,
we only need consider adding one variable at a time.

\begin{lemma}
\label{le:immediately-extends}
$\pi$ is a maximal $\modelsstar$-model of $\Gamma$ if and only if
$\pi\modelsstar \Gamma$ and
there exists no $\pi'$  minimally extending $\pi$
such that $\pi'\modelsstar \Gamma$.
\end{lemma}

\begin{proof}
%\commentphstartend
$\Rightarrow$ is obvious.
For the converse, assume that there exists no $\pi'$  minimally extending $\pi$
such that $\pi'\modelsstar \Gamma$;
 consider any $\pi''$ extending $\pi$.
Then there exists $\pi'$ that minimally extends $\pi$
and such that $\pi''$ extends or equals $\pi'$.
By the hypothesis, $\pi'\not\modelsstar\Gamma$,
which implies, using Lemma~\ref{le:modelsstar-basic}(ii),
that $\pi''\not\modelsstar\Gamma$, proving that
$\pi$ is a maximal $\modelsstar$-model of $\Gamma$.
\end{proof}

Starting with the empty model we  grow a
maximal $\modelsstar$-model of $\Gamma$,
%model $\modelsstar$-satisfying $\Gamma$,
by (iteratively) replacing the model with one minimally extending it and still $\modelsstar$-satisfying $\Gamma$,
if such a model exists. Otherwise, we have a maximal $\modelsstar$-model $\pi$ of $\Gamma$.
By Corollary~\ref{cor:max-modelsstar-consistent}, we can test if $\Gamma$ is consistent by checking $\pi\models\Gamma$.

\newcommand{\dd}{d}

%\commentph
We say that a \emph{satisfaction test} is a test of the form
$\pi \models \phi$ for some $\pi\in\Glex$ and $\phi\in\calL$;
a \emph{$\modelsstar$-satisfaction test} is a test of the form
$\pi\modelsstar \phi$.

%\commentph
Suppose that $\dd$ is an upper bound on the domain size of each variable,
i.e., for each $X\in V$, $|\und{X}| \le \dd$.
We first have to check that the empty model is a $\modelsstar$-model of $\Gamma$;
 this involves a number $|\Gamma|$ of $\modelsstar$-satisfaction tests.
At each iterative step, there are less than $n\dd!$ possible extensions $\pi$ of the current model $\pi'$ to an extra variable,
and there are at most $n$ iterative steps,
so there are less than $|\Gamma|n^2\dd!$ $\modelsstar$-satisfaction tests required in the whole iterative process.
Finally we have to check that the produced maximal $\modelsstar$-model $\pi$ of $\Gamma$ satisfies $\Gamma$,
which involves $|\Gamma|$ satisfaction tests.
In total we have at most $|\Gamma|$ satisfaction tests
and  $|\Gamma|n^2\dd!$ $\modelsstar$-satisfaction tests.

For bounded domain size, testing consistency of
strongly compositional
$\Gamma$
can be performed with polynomial number of
tests of the form
$\pi\models\phi$ or $\pi\modelsstar\phi$.
For this algorithm to be efficient we need that
 each preference statement
allows
efficient tests of $\pi\models\phi$ and $\pi\modelsstar\phi$.
In Section~\ref{sec:pref-languages-calLpq} we define languages that satisfy these properties.

%\commentph
\subsection{Equivalence w.r.t.~Compositional Statements}
\label{subsec:further-max-models}

%\commentph
The following result implies that
if $\Gamma$ is compositional then
two maximal models of $\Gamma$ involve exactly the same set of variables,
and if two outcomes are equivalent in one maximal model of $\Gamma$
then they are equivalent in all models of $\Gamma$.

%\commentph
\begin{proposition}
\label{pr:compos-max-models}
Assume that $\Gamma\subseteq\calL$ is consistent and compositional.
%Let $\pi,\pi'$ be maximal models of $\Gamma$.
%Then $V_\pi = V_{\pi'}$.
Let $\pi,\pi'$ be  models of $\Gamma$
with $\pi$ a maximal model of $\Gamma$.
Then, $V_\pi \supseteq V_{\pi'}$; and
 $V_\pi = V_{\pi'}$ if and only if $\pi'$ is a maximal model of $\Gamma$.
Also, if, for $\alpha,\beta\in\und{V}$,
we have $\alpha \equiv_\pi \beta$
then we have $\alpha \equiv_{\pi'} \beta$,
and in fact $\Gamma \models \alpha\equiv\beta$.
\end{proposition}

\begin{proof}
%\commentphstart
Let $\pi$ be a  maximal model of $\Gamma$
and let $\pi'$ be any model of $\Gamma$.
Compositionality of $\Gamma$ implies that
$\pi\circ\pi' \models\Gamma$,
and thus, $\pi\circ\pi' = \pi$,
since $\pi$ is a maximal model of $\Gamma$,
and $\pi\circ\pi'\sqsupseteq\pi$.
Therefore, $V_\pi \cup V_{\pi'} = V_\pi$,
and hence, $V_{\pi'} \subseteq V_\pi$.
If $\pi'$ is a maximal model of $\Gamma$ then the same argument implies that
$V_{\pi'} \supseteq V_\pi$ and thus $V_{\pi'} = V_\pi$.
Conversely, if $V_{\pi'} = V_\pi$ then $\pi'$ is a maximal model of $\Gamma$,
since otherwise $\pi'$ could be extended to a maximal model $\pi''$ of $\Gamma$,
which would imply $V_{\pi'}$ being a strict subset of $V_{\pi''} = V_\pi = V_{\pi'}$.

%\commentphend
Assume that $\alpha \equiv_\pi \beta$.
Then $\alpha(V_\pi) = \beta(V_\pi)$,
and so $\alpha(V_{\pi'}) = \beta(V_{\pi'})$
and $\pi' \models \alpha\equiv\beta$.
Since $\pi'$ is arbitrary,
we have $\Gamma\models\alpha\equiv\beta$.
\end{proof}
%
%\commentph
\noindent
Proposition~\ref{pr:compos-max-models} justifies the next definition,
and implies the following lemma.

%\commentph
\begin{definition}
For consistent and compositional $\Gamma\subseteq\calL$,
we write $V^\Gamma$ for the set of variables appearing in any maximal model of $\Gamma$,
i.e., $V^\Gamma = V_\pi$, for any maximal model $\pi$ of $\Gamma$.
\end{definition}

%\commentph
\begin{lemma}
\label{le:equiv-compos-V-Gamma}
Consider any consistent and compositional $\Gamma\subseteq\calL$,
and any $\alpha,\beta\in\und{V}$.
Then $\Gamma\models\alpha\equiv\beta$ $\iff$ $\alpha(V^\Gamma) = \beta(V^\Gamma)$.
\end{lemma}

\begin{proof}
%\commentphstart
Assume that $\alpha(V^\Gamma) \not= \beta(V^\Gamma)$.
Consider any maximal model $\pi$ of $\Gamma$.
Then,  $V_\pi = V^\Gamma$ so $\alpha$ and $\beta$ differ on some variable involved
in $\pi$, which implies that $\pi\not\models\alpha\equiv\beta$,
and thus, $\Gamma\not\models\alpha\equiv\beta$.

%\commentphend
To prove the converse, assume that $\alpha(V^\Gamma) = \beta(V^\Gamma)$.
Then for any $\pi\in\Glex$ with  $\pi\models\Gamma$,
we have $V_\pi \subseteq V^\Gamma$, by Proposition~\ref{pr:compos-max-models},
and thus, $\pi\models\alpha\equiv\beta$.
Therefore, $\Gamma\models\alpha\equiv\beta$.
\end{proof}

%\commentph
\subsection{Further Properties for Inference with
(Strongly) Compositional Statements}

%\commentph
When $\phi$ is strongly compositional
then there is a further simple composition property
just involving $\modelsstar$, as expressed by the following lemma.

%\commentph
\begin{lemma}
\label{le:modelsstar-basic-composition}
Let $\phi\in\calL$, and let $\pi,\pi'\in\Glex$.
If $\phi$ is strongly compositional then
$\pi\modelsstar\phi$ and $\pi'\modelsstar\phi$ implies $\pi\circ\pi'\modelsstar\phi$.
Also, if $\Gamma\subseteq\calL$ is strongly compositional then
$\pi\modelsstar\Gamma$ and $\pi'\modelsstar\Gamma$ implies $\pi\circ\pi'\modelsstar\Gamma$.
\end{lemma}

\begin{proof}
%\commentphstartend
Assume that $\pi\modelsstar\phi$ and $\pi'\modelsstar\phi$,
so there exists $\pi''$ with
$\pi''\sqsupseteq\pi'$ and $\pi''\models\phi$.
If $\phi$ is strongly compositional then
$\pi\circ\pi''\models\phi$, and thus,
$\pi\circ\pi'\modelsstar\phi$, since $\pi\circ\pi''\sqsupseteq\pi\circ\pi'$.
The second part follows immediately from the first.
\end{proof}

%\commentph
The following result
states that
if $\Gamma$ is strongly compositional
and if two outcomes are equivalent in one maximal $\modelsstar$-model of $\Gamma$
then they are equivalent in all $\modelsstar$-models of $\Gamma$.
Also, the maximal $\modelsstar$-models satisfy exactly the same elements of $\Gamma$.

%\commentph
\begin{proposition}
\label{pr:strong-compos-max-modelsstar}
Assume that $\Gamma\subseteq\calL$ is strongly compositional.
Let $\pi,\pi'$ be maximal $\modelsstar$-models of $\Gamma$.
Then, the following hold.
\begin{itemize}
 \item[(i)]  For every $\phi\in\Gamma$, $\pi\models\phi$ $\iff$ $\pi'\models\phi$.
 \item[(ii)] If $\Gamma$ is consistent then the set of maximal models of $\Gamma$
  is equal to the set of maximal $\modelsstar$-models of $\Gamma$.
\end{itemize}
\end{proposition}

\begin{proof}
%\commentphstart
(i):
Assume that $\pi'\models\phi$.
Since $\pi\modelsstar\Gamma$ we have $\pi\modelsstar\phi$,
and so $\pi\circ\pi'\models\phi$, since $\phi$ is strongly compositional.
Thus, $\pi\models\phi$, since $\pi\circ\pi' = \pi$.
%, e.g., using part (i).
Reversing the roles of $\pi$ and $\pi'$ in the argument, we have $\pi\models\phi$ $\iff$ $\pi'\models\phi$.

%\commentphend
(ii):
%If $\pi\models\Gamma$ then obviously $\Gamma$ is consistent.
%Conversely, suppose that $\Gamma$ is consistent,
Suppose that $\Gamma$ is consistent,
and so there exists some $\pi''$ with $\pi''\models\Gamma$.
Now, $\pi\modelsstar\Gamma$, so $\pi\circ\pi''\models\Gamma$,
since $\Gamma$ is strongly compositional.
Then $\pi\circ\pi''\modelsstar\Gamma$, and thus, $\pi\circ\pi''=\pi$,
since $\pi$ is a maximal $\modelsstar$-models of $\Gamma$.
So we have $\pi\models\Gamma$.
This argument also shows that $V_{\pi''} \subseteq V_\pi$,
which implies that $\pi$ is a maximal model of $\Gamma$,
using Proposition~\ref{pr:compos-max-models}.
We have just shown that, assuming $\Gamma$ is consistent,
every maximal $\modelsstar$-model of $\Gamma$
is a maximal model of $\Gamma$.
Also, by Proposition~\ref{pr:compos-max-models},
 every maximal model $\tau$ of $\Gamma$ has $V_\tau = V_\pi$.
We have, by Lemma~\ref{le:modelsstar-basic}(i), that $\tau\modelsstar\Gamma$.
Let $\tau'$ be a maximal $\modelsstar$-model of $\Gamma$
with $\tau'\sqsupseteq\tau$.
By the previous argument $\tau'$ is also a maximal model of $\Gamma$,
which implies, using Proposition~\ref{pr:compos-max-models},
that $V_{\tau'}= V_\pi = V_\tau$, and so $\tau' = \tau$,
and thus, $\tau$ is a maximal $\modelsstar$-model of $\Gamma$, as required.
\end{proof}

\newcommand{\modelsmax}{\models^{max}}

%\commentph
One can define a notion of max-model inference,
which is characterised by the  lemma below.
We define $\pi \modelsmax \Gamma$ if
$\pi$ is a maximal model of $\Gamma$.
We also define $\Gamma \modelsmax \phi$ if
$\pi\models\phi$ for every maximal model $\pi$ of $\Gamma$.

\newcommand{\Gammamax}{\Gamma_{max}}

%\commentph
\begin{lemma}
\label{le:modelsmax}
Let $\Gamma\cup\set{\phi,\neg\phi} \subseteq\calL$
and suppose that $\Gamma\cup\set{\neg\phi}$  is compositional.
If $\Gamma \models\phi$ then $\Gamma\modelsmax\phi$.
Now suppose that $\Gamma \not\models\phi$ and let
$\pi$ be any maximal model of $\Gamma\cup\set{\neg\phi}$.
Then, $\Gamma\modelsmax\phi$ $\iff$ $\pi\not\modelsmax\Gamma$,
which holds if and only if $V^{\Gamma\cup\set{\neg\phi}}\not=V^\Gamma$.
\end{lemma}

\begin{proof}
%\commentphstart
The definitions immediately imply that
if $\Gamma \models\phi$ then $\Gamma\modelsmax\phi$.
Now assume that $\Gamma \not\models\phi$,
and so $\Gamma\cup\set{\neg\phi}$ is consistent,
and let $\pi$ be any maximal model of $\Gamma\cup\set{\neg\phi}$.
Suppose that $\Gamma\modelsmax\phi$.
The fact that $\pi\not\models\phi$ implies that $\pi\not\modelsmax\Gamma$.
Conversely, assume that $\Gamma\not\modelsmax\phi$
so there exists $\pi'\in\Glex$ such that
$\pi'\modelsmax\Gamma$ and $\pi'\models\neg\phi$,
and thus $\pi'\models\Gamma\cup\set{\neg\phi}$.
 Proposition~\ref{pr:compos-max-models} implies
that $V_{\pi'}\subseteq V_\pi$.
But since $\pi\models\Gamma$ and $\pi'\modelsmax\Gamma$
we have also $V_{\pi}\subseteq V_{\pi'}$ and thus $V_{\pi'}= V_\pi$.
Proposition~\ref{pr:compos-max-models} then implies that
$\pi\modelsmax\Gamma$.

%\commentphend
We have $V^{\Gamma\cup\set{\neg\phi}} = V_\pi \subseteq V^\Gamma$.
Proposition~\ref{pr:compos-max-models} implies that
  $\pi\modelsmax\Gamma$ if and only if $V_\pi = V^\Gamma$,
  which is if and only if  $V^{\Gamma\cup\set{\neg\phi}}=V^\Gamma$.
\end{proof}

%\commentph
\subsection{Decreasing Preference Statements and Sets of Models}
\label{subsec:decreasing}

%\commentph
In this section we show some results that are useful in proving
that certain preference statements are strongly compositional.

%\commentph
Let us say that $\phi\in\calL$ is \emph{decreasing}  if
for all $\pi,\pi'\in\Glex$ with $\pi'$ extending $\pi$,
we have $\pi'\models\phi$ $\Rightarrow$ $\pi\models\phi$.
It follows easily from the definitions that
if $\phi$ is decreasing, then, for all $\pi\in\Glex$,
$\pi \modelsstar \phi$ $\iff$ $\pi\models\phi$.
This leads to the following result.

%\commentph
\begin{lemma}
\label{le:compos-strong-iff-compos}
Let $\phi\in\calL$
be decreasing.
%Suppose that $\phi$ is such that for all $\pi,\pi'\in\Glex$ with $\pi'$ extending $\pi$,
%we have $\pi'\models\phi$ $\Rightarrow$ $\pi\models\phi$.
Then $\phi$ is strongly compositional if and only if $\phi$ is compositional.
\end{lemma}

%\commentph
For example, the non-strict statements of the flight connections
examples from the earlier sections are decreasing, while
the strict statements in these examples are not.

\newcommand{\phibar}{\overline{\phi}}

%\commentph
For $\calM\subseteq\Glex$, we say that
$\calM$ is \emph{decreasing} if, for any $\pi,\pi'\in\Glex$
such that $\pi'$ extends $\pi$, we have
$\pi' \in\calM$ $\Rightarrow$ $\pi\in\calM$.

%\commentph
We say that $\calM$ \emph{contains all models of} $\phi$
if, for all $\pi\in\Glex$,
$\pi\models\phi$ $\Rightarrow$ $\calM\ni\pi$.

%\commentph
The following result is helpful for proving that
a preference statement $\phi$ is strongly compositional.

%\commentph
\begin{proposition}
\label{pr:strong-compos-equiv-form}
Let $\phi\in\calL$, and assume that $\phi$ is consistent.
Let $\calM_\phi$ be a subset of $\Glex$.
The following two conditions are equivalent:
\begin{itemize}
  \item[(I)] $\calM_\phi$ is decreasing and contains all models of $\phi$,
  and
  for all $\pi,\pi'\in\Glex$, if $\pi\in\calM_\phi$ and $\pi'\models\phi$ then $\pi\circ\pi'\models\phi$.
  \item[(II)] $\phi$ is strongly compositional, and
  for all $\pi\in\Glex$, $\pi\modelsstar\phi$ $\iff$ $\pi\in\calM_\phi$.
\end{itemize}
\end{proposition}

\begin{proof}
%\commentphstart
(I)$\Rightarrow$(II):
Assume (I).
First, let us assume that $\pi\in\calM_\phi$.
Since $\phi$ is consistent, there exists $\pi'$ with $\pi'\models\phi$,
and so (I) implies that $\pi\circ\pi'\models\phi$,
and thus, $\pi\modelsstar\phi$, since $\pi\circ\pi'\sqsupseteq\pi$.
For proving the converse let us now assume that $\pi\modelsstar\phi$,
so there exists $\pi'\in\Glex$ with $\pi'\sqsupseteq\pi$ and $\pi'\models\phi$.
Thus, $\pi'\in\calM_\phi$, and because $\calM_\phi$ is decreasing, we then have
 $\pi\in\calM_\phi$.
We have shown that for all $\pi\in\Glex$, $\pi\modelsstar\phi$ $\iff$ $\pi\in\calM_\phi$.
(I) then also implies that $\phi$ is strongly compositional.

%\commentphend
(II)$\Rightarrow$(I):
Lemma~\ref{le:modelsstar-basic}(i) implies that
$\calM_\phi$ contains all models of  $\phi$,
and Lemma~\ref{le:modelsstar-basic}(ii) implies that
$\calM_\phi$ is decreasing.
The fact that $\phi$ is strongly compositional then implies (I).
\end{proof}

%\commentph
For $\phi,\phi'\in\calL$,
we say that $\phi'$ \emph{is a relaxation of $\phi$} if
$\phi\models\phi'$, i.e.,
for all $\pi\in\calG$,
$\pi\models\phi$ $\Rightarrow$ $\pi\models\phi'$.

%\commentph
We have the following special case of Proposition~\ref{pr:strong-compos-equiv-form}.

%\commentph
\begin{proposition}
\label{pr:strong-compos-equiv-form-phibar}
Let $\phi, \phibar\in\calL$, and assume that $\phi$ is consistent.
The following two conditions are equivalent:
\begin{itemize}
  \item[(I)] $\phibar$ is a decreasing relaxation of $\phi$ such that
  for all $\pi,\pi'\in\Glex$, if $\pi\models\phibar$ and $\pi'\models\phi$ then $\pi\circ\pi'\models\phi$.
  \item[(II)] $\phi$ is strongly compositional, and
  for all $\pi\in\Glex$, $\pi\modelsstar\phi$ $\iff$ $\pi\models\phibar$.
\end{itemize}
\end{proposition}

\begin{proof}
%\commentphstartend
Define $\calM_\phi$ to be
all $\pi\in\Glex$ such that $\pi\models\phibar$.
Note that (I) holds if and only if
$\calM_\phi$ is decreasing and contains all models of $\phi$,
  and  if $\pi\in\calM_\phi$ and $\pi'\models\phi$ then $\pi\circ\pi'\models\phi$.
Proposition~\ref{pr:strong-compos-equiv-form} implies that (I) $\iff$ (II).
\end{proof}

%\commentph
\subsection{Important Instances of Strongly Compositional Preference Statements}
\label{sec:Important-Instances}

%\commentph
\subsubsection*{Projections to $Y$}

%\commentph
\noindent
Our computational techniques can be expressed in terms of projections of preference statements to a single variable.

%\commentph
Let $\RR \subseteq \outc\times\outc$, let
$Y \in V$ be a variable, and let $A\subseteq V-\set{Y}$ be a set of variables not containing $Y$.
Define $\RR^{\downarrow Y}$,
the projection of $\RR$ to $Y$, to be $\set{(\alpha(Y), \beta(Y)) \st \alphabeta \in \RR}$.
Also, define,
$\RR_A^{\downarrow Y}$, the $A$-restricted projection to $Y$,
to be the set of pairs $(\alpha(Y), \beta(Y))$
such that $\alphabeta \in \RR$ and $\alpha(A) = \beta(A)$.
$\RR_A^{\downarrow Y}$ is the projection to $Y$ of pairs that agree on $A$.
Thus, $\RR^{\downarrow Y} = \RR_\emp^{\downarrow Y}$.

%\commentph
From \cite{Wilson14} we have (a variation of the following):

%\commentph
\begin{lemma}
\label{le:alpha-pi-beta-Gamma-sequential}
Consider any lexicographic model $\pi\in\Glex$, written as
$(Y_1, \ge_1), \ldots, (Y_k, \ge_{k})$.
where  each $Y_i$, for $i=1, \ldots, k$, is a variable in $V$,
and each $\ge_{Y_i}$ is a total order on $\und{Y_i}$.
For $i=1, \ldots, k$, define
$A_i$ to be the set of earlier variables than $Y_i$, i.e.,
$A_i = \set{Y_1, \ldots, Y_{i-1}}$.
Let
$\RR\subseteq\und{V}\times\und{V}$.
Then
$\succceq_\pi \supseteq\, \RR$ if and only if
for all $i=1, \ldots, k$, $\ge_{i}\, \supseteq  \RR_{A_i}^{\downarrow Y_i}$.
\end{lemma}

%\commentph
Note that, by Lemma~\ref{le:extends-pi-basic},
$\alpha\ge\beta$ is a decreasing element.
Lemma~\ref{le:composition-basic} implies that
$\alpha\ge\beta$ is compositional, and thus, strongly compositional,
by Lemma~\ref{le:compos-strong-iff-compos}.
If $\alpha>\beta$ is inconsistent (so that $\alpha=\beta$)
then it is trivially strongly compositional.
Otherwise, we can use Proposition~\ref{pr:strong-compos-equiv-form-phibar}
with $\phi = \alpha>\beta$ and $\phibar = \alpha\ge\beta$
to show, using Lemma~\ref{le:composition-basic}, that
$\alpha>\beta$ is also strongly compositional
and $\pi\modelsstar\alpha>\beta$
$\iff$ $\pi\models\alpha\ge\beta$.
We therefore have:

%\commentph
\begin{lemma}
\label{le:strong-compos-alpha-beta}
For any outcomes, $\alpha,\beta\in\und{V}$,
statements $\alpha\ge\beta$ and $\alpha>\beta$
are strongly compositional,
and, if $\alpha\not=\beta$, for $\pi\in\Glex$,
$\pi\modelsstar\alpha>\beta$ $\iff$ $\pi\modelsstar\alpha\ge\beta$
$\iff$ $\pi\models\alpha\ge\beta$.
\end{lemma}

%\commentph
Let $\RR$ be a subset of $\und{V}\times\und{V}$,
and let $\phi^\RR$ be some statement satisfying:
$\pi\models\phi^\RR$ if and only if
$\succceq_\pi \supseteq\RR$.

%\commentph
If $\pi'$ extends $\pi$ then $\succceq_{\pi}$ $\supseteq$ $\succceq_{\pi'}$,
which implies that
$\phi^\RR$ is decreasing.
For any lexicographic model $\pi$ we have
$\pi\models\phi^\RR$ if and only if
for all $(\alpha,\beta)\in\RR$,
$\pi\models\alpha\ge\beta$,
which implies that $\phi^\RR$ is strongly compositional,
by Lemmas~\ref{le:strong-compos-alpha-beta} and~\ref{le:compos-conjunction}.
We therefore have:

%\commentph
\begin{proposition}
\label{pr:phi-calR-strongly-compos}
For any $\RR\subseteq\und{V}\times\und{V}$,
the preference statement $\phi^\RR$ is strongly compositional,
and, for any $\pi\in\Glex$, we have $\pi\modelsstar\phi^\RR$ $\iff$
$\pi\models\phi^\RR$
\end{proposition}

%\commentph
We say that $\psi$ is a \emph{strict version of $\phi^\RR$}
if $\phi^\RR$ is a relaxation of $\psi$,
and $\psi$
satisfies the following monotonicity property regarding the
strict preferences among $\RR$:
for all $\pi,\pi'\in\Glex$,
if $\pi\models\psi$ and
$\pi'\models\phi^\RR$ and $\succ_{\pi'}\supseteq (\succ_\pi \cap \RR)$
then $\pi' \models\psi$.

%\commentph
There are many strict versions of $\phi^\RR$ (unless $\RR$ is very small).
We give two simple examples $\psi_1$ and $\psi_2$ of
strict versions of $\phi^\RR$.
Let $\psi_1$ be such that
$\pi\models\psi_1$ if and only if $\succ_\pi$ $\supseteq \RR$.
Let $\psi_2$ be such that
$\pi\models\psi_2$ if and only if $\succceq_\pi$ $\supseteq \RR$
and there exists some $(\alpha,\beta)\in\RR$ such that
$\alpha \succ_\pi \beta$.

%\commentph
\begin{proposition}
\label{pr:phi-calR-strict-strongly-compos}
Let $\RR\subseteq\und{V}\times\und{V}$,
and suppose that $\psi$ is a strict version of $\phi^\RR$.
Then $\psi$ is strongly compositional, and for $\pi\in\Glex$,
$\pi\modelsstar\psi$ $\iff$ $\pi\models\phi^\RR$.
\end{proposition}

\begin{proof}
%\commentphstartend
We have that $\phi^\RR$ is a decreasing relaxation of $\psi$.
It is sufficient to prove that
for any $\pi,\pi'\in\Glex$,
if $\pi\models\phi^\RR$ and $\pi'\models\psi$ then $\pi\circ\pi'\models\psi$,
since then we can use Proposition~\ref{pr:strong-compos-equiv-form-phibar}
to prove the result.
Assume that $\pi\models\phi^\RR$ and $\pi'\models\psi$.
Since $\phi^\RR$ is a relaxation of $\psi$,
we have $\pi'\models \phi^\RR$, and thus,
$\pi\circ\pi'\models\phi^\RR$,
since $\phi^\RR$ is compositional, by Proposition~\ref{pr:phi-calR-strongly-compos}.
Consider any $(\alpha,\beta)\in\RR$ such that
$\alpha \succ_{\pi'} \beta$.
We also have $\alpha\succceq_{\pi}\beta$, because $\pi\models\phi^\RR$,
and so, $\alpha \succ_{\pi\circ\pi'} \beta$,
using Lemma~\ref{le:composition-basic}.
We have shown that $\succ_{\pi\circ\pi'}$ $\supseteq (\succ_{\pi'} \cap \RR)$,
which, since $\psi$ is a strict version of $\phi^\RR$, implies $\pi\circ\pi' \models\psi$, as required.
\end{proof}

%\commentph
\subsubsection{Negated Statements}
\label{subsec:negated}

%\commentph
\noindent
Suppose that $\pi\models\neg\phi^\RR$, i.e., $\pi\not\models\phi^\RR$,
and consider any $\pi'\in\Glex$.
Then, $\pi\circ\pi'$ extends or equals $\pi$,
which implies that $\pi\circ\pi'\not\models\phi^\RR$,
since $\phi^\RR$ is decreasing.
Thus, we have:

%\commentph
\begin{proposition}
\label{pr:neg-phi-calR-compos}
Preference statement $\neg\phi^\RR$ is compositional
for any $\RR\subseteq\und{V}\times\und{V}$.
\end{proposition}

%\commentph
\begin{lemma}
\label{le:phi-calR-notequal}
Let $\RR\subseteq\und{V}\times\und{V}$,
and let $\pi,\pi'\in\Glex$ be such that $\pi'$ extends $\pi$.
Suppose that for all $(\alpha,\beta)\in\RR$
there exists $X\in V_\pi$ such that $\alpha(X)\not=\beta(X)$.
Then, for any $(\alpha,\beta)\in\RR$,
$\alpha \succceq_\pi \beta$
%$\iff$ $\alpha \succ_\pi \beta$
$\iff$ $\alpha \succceq_{\pi'} \beta$.
Also,
$\pi\models\phi^\RR$ $\iff$ $\pi'\models\phi^\RR$.
\end{lemma}

\begin{proof}
%\commentphstartend
Consider any $(\alpha,\beta)\in\RR$.
Let $Y$ be the earliest variable in $V_\pi$ such that
$\alpha(Y)\not=\beta(Y)$
(this is well-defined, by the hypothesis).
Then, $\alpha \succceq_\pi \beta$
%$\iff$ $\alpha \succ_\pi \beta$
$\iff$ $\alpha \succceq_{\pi'} \beta$
$\iff$ $\alpha(Y) >_Y \beta(Y)$, where
$>_Y$ is the
strict part of the
 ordering for $Y$ in $\pi$.
We then have $\pi\models\phi^\RR$ if and only if
for all $(\alpha,\beta)\in\RR$, $\alpha \succceq_\pi \beta$
if and only if for all $(\alpha,\beta)\in\RR$, $\alpha \succceq_{\pi'} \beta$
if and only if
$\pi'\models\phi^\RR$.
\end{proof}
%
%\commentph
\begin{lemma}
\label{le:phi-calR-pi-equiv}
Let $\RR\subseteq\und{V}\times\und{V}$,
and let $\pi,\pi'\in\Glex$.
Suppose that
%for all $X\in V_\pi$, we have $\RR^{\downarrow X} \subseteq \ =$
 for all $(\alpha,\beta)\in\RR$ we have $\alpha(V_\pi) = \beta(V_\pi)$.
%We have $\pi\models\phi^\RR$.
Then, $\succceq_{\pi'} \cap\ \RR = \ \succceq_{\pi\circ\pi'} \cap\ \RR$,
and thus,
$\pi'\models\phi^\RR\ \iff\  \pi\circ\pi'\models\phi^\RR$.
\end{lemma}

\begin{proof}
%\commentphstartend
Consider any $(\alpha,\beta)\in\RR$.
We have $\alpha\equiv_\pi \beta$.
By Lemma~\ref{le:composition-basic},
$\alpha\succceq_{\pi'} \beta$ $\iff$ $\alpha\succceq_{\pi\circ\pi'} \beta$.
This implies $\succceq_{\pi'} \cap\ \RR = \ \succceq_{\pi\circ\pi'} \cap\ \RR$,
and $\pi'\models\phi^\RR\ \iff\  \pi\circ\pi'\models\phi^\RR$.
\end{proof}
%
%%\commentph
%\begin{lemma}
%\label{le:neg-calR-not-acyclic}
%Let $\RR\subseteq\und{V}\times\und{V}$,
%and let $\pi,\pi'\in\Glex$ be such that $\pi'$ extends $\pi$.
%Suppose that $X$ is the first variable in $\pi$ on which
%some pair in $\RR$ differs,
%so that there exists $(\alpha,\beta)\in\RR$ such that
%$\alpha(X)\not=\beta(X)$, and this does not hold for any earlier variable in $\pi$.
%If $\RR^{\downarrow X}$ is not acyclic then
%$\pi\not\models\phi^\RR$ and $\pi'\not\models\phi^\RR$.
%\end{lemma}
%
%\begin{proof}
%%\commentphstartend
%Let $\ge_X$ be the non-strict relation for $X$ in $\pi$.
%If $\pi\models\phi^\RR$ or $\pi'\models\phi^\RR$ then,
% by Lemma~\ref{le:alpha-pi-beta-Gamma-sequential},
% $\ge_X$ contains $\RR^{\downarrow X}$
%and thus, the latter is acyclic.
%\end{proof}
%\commentph
\begin{lemma}
\label{le:neg-calR-not-acyclic}
Let $\RR\subseteq\und{V}\times\und{V}$,
and let $\pi\in\Glex$.
Suppose that $X$ is the first variable in $\pi$ on which
some pair in $\RR$ differs,
so that there exists $(\alpha,\beta)\in\RR$ such that
$\alpha(X)\not=\beta(X)$, and this does not hold for any earlier variable in $\pi$.
If $\RR^{\downarrow X}$ is not acyclic then
$\pi\not\models\phi^\RR$.
\end{lemma}

\begin{proof}
%\commentphstartend
Let $\ge_X$ be the non-strict relation for $X$ in $\pi$.
If $\pi\models\phi^\RR$ then,
 by Lemma~\ref{le:alpha-pi-beta-Gamma-sequential},
 $\ge_X$ contains $\RR^{\downarrow X}$
and thus, the latter is acyclic.
\end{proof}

%\commentph
Define $\calM_{\neg \phi^\RR}$ to be the set of models $\pi\in\Glex$
such that either (i) $\pi\not\models\phi^\RR$ or
(ii) for every variable $X\in V_\pi$,
if $\RR^{\downarrow X}$ is acyclic then $\RR^{\downarrow X} \subseteq\ =$.

%\commentph
Define $\calM'_{\neg \phi^\RR}$ to be the set of models $\pi\in\Glex$
such that either (i) $\pi\not\models\phi^\RR$ or
(ii) there is no variable $X\in V_\pi$
such that  $\RR^{\downarrow X}$ is acyclic and irreflexive.

%\commentph
Note that $\calM_{\neg \phi^\RR} \subseteq \calM'_{\neg \phi^\RR}$.

%\commentph
\begin{lemma}
\label{le:neg-phi-calR-calM-compos}
If $\pi\in\calM_{\neg \phi^\RR}$ and
$\pi'\models\neg\phi^\RR$ then $\pi\circ\pi'\models\neg\phi^\RR$.
\end{lemma}

\begin{proof}
%\commentphstartend
Suppose that $\pi\in\calM_{\neg \phi^\RR}$ and
$\pi'\models\neg\phi^\RR$.
First consider the case when $\pi\not\models\phi^\RR$.
The fact that $\phi^\RR$ is decreasing implies that $\pi\circ\pi'\not\models\phi^\RR$,
i.e., $\pi\circ\pi'\models\neg\phi^\RR$.
Now consider the other case, when $\pi\models\phi^\RR$.
If for all $(\alpha,\beta)\in\RR$ we have $\alpha(V_\pi) = \beta(V_\pi)$
then Lemma~\ref{le:phi-calR-pi-equiv} implies that $\pi\circ\pi'\models\neg\phi^\RR$.
Otherwise, let $X$ be the first variable in $\pi$ on which
some pair in $\RR$ differ.
The definition of $\calM_{\neg \phi^\RR}$ implies that
$\RR^{\downarrow X}$ is not acyclic, and thus,
$\pi\circ\pi'\models\neg\phi^\RR$, by Lemma~\ref{le:neg-calR-not-acyclic}.
\end{proof}

%\commentph
\begin{lemma}
\label{le:calM-prime-decreasing}
$\calM'_{\neg \phi^\RR}$ is decreasing.
\end{lemma}

\begin{proof}
%\commentphstartend
Suppose that $\pi'$ extends $\pi$
and that $\pi' \in \calM'_{\neg \phi^\RR}$.
We need to show that $\pi \in \calM'_{\neg \phi^\RR}$.
We proceed using proof by contradiction.
Assume that $\pi \notin \calM'_{\neg \phi^\RR}$.
Thus, $\pi\models\phi^\RR$
and there exists variable $X\in V_\pi$
such that  $\RR^{\downarrow X}$ is acyclic and irreflexive.
Since $\pi' \in \calM'_{\neg \phi^\RR}$ and $V_{\pi'}\supseteq V_\pi$,
we must have $\pi'\not\models\phi^\RR$.
Because $\RR^{\downarrow X}$ is irreflexive,
for all  $(\alpha,\beta)\in\RR$,
$\alpha$ and $\beta$ differ on variable $X$.
 Lemma~\ref{le:phi-calR-notequal} leads to the required contradiction.
\end{proof}

%\commentph
\begin{definition}
\label{def:simult-decisive}
For $\RR\subseteq\und{V}\times\und{V}$,
we say that $\RR$ is
\emph{simultaneously decisive} if
for all $X\in V$:
if $\RR^{\downarrow X}$ is acyclic then either
$\RR^{\downarrow X}$ is irreflexive or $\RR^{\downarrow X} \subseteq\ =$.
\end{definition}

%\commentph
\begin{proposition}
\label{pr:neg-strongly-compos}
If $\RR\subseteq\und{V}\times\und{V}$
is simultaneously decisive then
$\neg\phi^\RR$ is strongly compositional,
and for all $\pi\in\Glex$, $\pi\modelsstar\neg\phi^\RR$ $\iff$ $\pi\in\calM_{\neg\phi^\RR}$.
\end{proposition}

\begin{proof}
%\commentphstart
We will first show that $\calM_{\neg \phi^\RR} = \calM'_{\neg \phi^\RR}$.
Suppose that $\pi\in\calM'_{\neg \phi^\RR}$.
If $\pi\not\models\phi^\RR$ then we clearly have $\pi\in\calM_{\neg \phi^\RR}$,
so we can assume that there is no variable $X\in V_\pi$
such that  $\RR^{\downarrow X}$ is acyclic and irreflexive.
Thus, by our assumption on $\RR$,
if $\RR^{\downarrow X}$ is acyclic then  $\RR^{\downarrow X} \subseteq\ =$,
and thus, $\pi\in\calM_{\neg \phi^\RR}$.

%\commentphend
Lemma~\ref{le:calM-prime-decreasing} then implies that $\calM_{\neg \phi^\RR}$
is decreasing, and it contains all models of $\neg \phi^\RR$.
Then, Lemma~\ref{le:neg-phi-calR-calM-compos} and
Proposition~\ref{pr:strong-compos-equiv-form} imply the result.
\end{proof}

\section{Preference Languages $\calLpq$ and $\calLpqn$}
\label{sec:pref-languages-calLpq}

Here we show that certain relatively expressive compact preference languages are strongly compositional.
This includes forms of the statements
$\phi^\RR$
from Proposition~\ref{pr:strong-compos-alpha-beta-calR},
where $\RR$ is a set of pairs of outcomes.
In many natural situations, $\RR$ can be exponentially large;
in the languages discussed here,
we are able to express certain exponentially large sets $\RR$ compactly.

\subsection{The Language $\calLpq$ }
\label{subsec:language-calLpq}

We will consider preference statements
of the
form $p \rhd q$ \eqq~$T$,
where
$\rhd$ is either $\ge$, or $\gg$ or $>$, and
 $P$, $Q$ and $T$ are subsets of $V$,
 with
$(P\cup Q) \cap T = \emp$,
and  $p \in \und{P}$ is an assignment to $P$,
and $q \in \und{Q}$ is an assignment to $Q$.
The statement $p \rhd q$ \eqq~$T$
represents that
$p$ is preferred to $q$ if $T$ is held constant, i.e.,
any outcome $\alpha$ extending $p$ is preferred to
any outcome $\beta$ that extends $q$ and agrees with $\alpha$ on variables $T$.
Formally, the semantics of this statement
relates to the set
$\phi^*$ which is defined to be the set of
pairs $\alphabeta$ of outcomes such that
$\alpha$ extends $p$, and $\beta$ extends $q$,
and $\alpha$ and $\beta$ agree on $T$,
i.e., $\alpha(T) = \beta(T)$.

Statements of the form $p \ge q$ \eqq~$T$
are called \emph{non-strict};
statements of the form $p \gg q$ \eqq~$T$,
are called \emph{fully strict},
and statements of the form
$p > q$ \eqq~$T$ are called \emph{weakly strict}.

Let $\calLpq$ be the set of all preference statements $\phi$
of the form $p \rhd q$ \eqq~$T$, as defined above.
For any statement $\phi\in\calLpq$ equalling $p \rhd q$ \eqq~$T$,
we define $\phidestrict$ to be
$p \ge q$ \eqq~$T$,
the \emph{non-strict version of $\phi$}.
For lex model $\pi$, we define:
\begin{itemize}
  \item $\pi$  satisfies $\phidestrict$ if
 $\alpha \succceq_\pi \beta$ for all $(\alpha,\beta)\in\phi^*$.
  \item $\pi$ satisfies fully strict $\phi$
if $\alpha \succ_\pi \beta$ for all $(\alpha,\beta)\in\phi^*$.
  \item $\pi$ satisfies weakly strict $\phi$ if
  $\pi$ satisfies $\phidestrict$ and
if $\alpha \succ_\pi \beta$ for some $(\alpha,\beta)\in\phi^*$.
\end{itemize}

For outcomes $\alpha$ and $\beta$, a non-strict preference of $\alpha$ over $\beta$
can be represented as
$\alpha \ge \beta$ \eqq~$\emp$,
which is equivalent to the preference statement $\gralphabeta$ introduced in Section~\ref{sec:lex-inf-strong-compos},
so we abbreviate it to that.
 Similarly, we abbreviate $\alpha > \beta$ \eqq~$\emp$ to $\alpha>\beta$
(which is also equivalent to  $\alpha \gg \beta$ \eqq~$\emp$).

\newcommand{\WW}{W}

We can write a statement $\phi\in\calLpq$
as
$ur \rhd us $ \eqq~$T$, where $u \in\und{U}$, $r\in\und{R}$,
$s \in \und{S}$, and $U$, $T$ and $R \cup S$ are (possibly empty) mutually disjoint subsets of $V$,
and for all $X \in R \cap S$, $r(X) \not= s(X)$.
For such a representation, we write
$u_\phi = u$, $r_\phi = r$, $s_\phi = s$, $U_\phi = U$, $R_\phi = R$, $S_\phi = S$ and $T_\phi = T$.
We assume, without loss of generality,  that
if $|\und{X}| = 1$ then $X\in T_\phi$.
This ensures that such a representation is unique.
We also define
$\WW_\phi = V - (R_\phi \cup S_\phi \cup T_\phi \cup U_\phi)$.

As discussed in \cite{Wilson09} (which, however, just considers non-strict statements), this is a relatively expressive preference language.
As well as preferences between outcomes, preferences between partial tuples can be expressed.
\emph{Ceteris paribus} statements can be expressed by using
$\phi$ with $\WW_\phi = \emp$.
More generally, any variables in $\WW_\phi$ are forced to be less important than variables in $R_\phi$ and $S_\phi$.

%\commentph
\subsubsection*{Projections to $Y$}

%\commentph
For comparative preference statement $\phi$
we abbreviate $(\phi^*)_\AA^{\downarrow Y}$ to $\phi_\AA^{\downarrow Y}$.

%\commentph
Proposition~1 of~\cite{Wilson14} leads to the following result.

%\commentph
\begin{proposition} \label{pr:rs-local-ordering}
Consider any element $\phi \in \calLpq$.
Let $\AA$ be a set of variables and let $Y$ be a variable not in $A$.
If $R_\phi \cap S_\phi \cap A \not=\emp$
then
$\phi_\AA^{\downarrow Y}$ is empty.
Otherwise,
$\phi_\AA^{\downarrow Y}$ consists of all pairs $(y, y') \in\und{Y}\times\und{Y}$
such that
(i) $y = y'$ if $Y \in T_\phi$;
(ii) $y = y' = u_\phi(Y)$ if $Y \in U_\phi$
(iii) $y = r_\phi(Y)$  if $Y \in R_\phi$; and
(iv) $y' = s_\phi(Y)$ if $Y \in S_\phi$.
Thus if $R_\phi \cap S_\phi \cap A =\emp$  and
$Y\in \WW_\phi$ then  $\phi_\AA^{\downarrow Y} =\und{Y}\times\und{Y}$.
\end{proposition}

%\commentph
The following lemma will be used later.

%\commentph
\begin{lemma}
\label{le:equality-for-strict}
Consider any $\phi\in\calLpq$, and any set of variables $\AA \subseteq V$.
We have the following.
\begin{itemize}
  \item[(i)] There exists $(\alpha,\beta)\in\phi^*$ such that $\alpha(\AA) = \beta(\AA)$
  if and only if $R_\phi\cap S_\phi \cap \AA = \emp$.
  \item[(ii)] $\alpha(\AA) = \beta(\AA)$  holds for all $(\alpha,\beta)\in\phi^*$
  if and only if
%  $(R_\phi \cup S_\phi \cup \WW_\phi) \cap \AA =\emp$.
$\AA \subseteq T_\phi \cup U_\phi$.
\end{itemize}
\end{lemma}

\begin{proof}
%\commentphstart
(i) First suppose that $R_\phi\cap S_\phi \cap \AA \not= \emp$,
choose some $X \in R_\phi\cap S_\phi \cap \AA$,
and consider any $(\alpha,\beta)\in\phi^*$.
Then $\alpha(X) = r_\phi(X)$ and $\beta(X)= s_\phi(X)\not=\alpha(X)$,
which shows that $\alpha(\AA) \not=\beta(\AA)$.
Conversely, suppose that $R_\phi\cap S_\phi \cap \AA = \emp$.
Let $s'_\phi$ be $s_\phi$ restricted to $S_\phi - R_\phi$.
Let $\alpha$ be any outcome extending $u_\phi$ and $r_\phi$ and $s'_\phi$.
Define $\beta$ by $\beta(X) = s_\phi(X)$ if $X\in S_\phi$,
and $\beta(X) = \alpha(X)$, otherwise.
Then $(\alpha,\beta)\in\phi^*$ and $\alpha(\AA) = \beta(\AA)$,
since $\alpha$ and $\beta$ differ only on $R_\phi \cap S_\phi$,
which is disjoint from $\AA$.

%\commentph
\noindent\ssk
(ii) Assume first that
$\AA \not\subseteq T_\phi \cup U_\phi$, i.e.,
$(R_\phi \cup S_\phi \cup \WW_\phi) \cap \AA \not=\emp$.
We will construct $\alpha$ and $\beta$ such that
 $\alpha(\AA) \not= \beta(\AA)$  and $(\alpha,\beta)\in\phi^*$.
Let $\alpha$ be any outcome extending $u_\phi$ and $r_\phi$
and such that $\alpha(X) \not=s_\phi(X)$ for all $X\in S_\phi$.
Let $\beta$ be any outcome extending $u_\phi$ and $s_\phi$ and $\alpha(T_\phi)$
and such that $\beta(X) \not = r_\phi(X)$ for all $X\in R_\phi$,
and also $\beta(X) \not= \alpha(X)$ for all $X\in \WW_\phi$
(we can do this because each element of the domain of each variable in $R_\phi \cup S_\phi \cup \WW_\phi$ includes at least two elements).
Then $(\alpha,\beta)\in\phi^*$ and
$\alpha(X) \not=\beta(X)$ for all $X \in R_\phi \cup S_\phi \cup \WW_\phi$,
which implies that $\alpha(\AA) \not= \beta(\AA)$.

%\commentphend
To prove the converse, assume that
$\AA \subseteq T_\phi \cup U_\phi$.
and consider any $(\alpha,\beta)\in\phi^*$.
Then, for each
$X\in T_\phi \cup U_\phi$,
we have $\alpha(X) = \beta(X)$, and so $\alpha(\AA) = \beta(\AA)$.
\end{proof}

%\commentph
\begin{lemma}
\label{le:phi-W-variable}
Let $\phi\in\calLpq$ and $\pi\in\Glex$.
Suppose that $\pi \models \phidestrict$, i.e.,
$\succceq_\pi$ $\supseteq$ $\phi^*$.
If $\WW_\phi \cap V_\pi \not=\emp$ then there exists
$X\in R_\phi\cap S_\phi\cap V_\pi$ that appears earlier in
$\pi$ than any variable in $\WW_\phi$.
\end{lemma}

\begin{proof}
%\commentphstart
Suppose otherwise, and let $X$ be the first variable in
$\WW_\phi$ that appears in $\pi$, and let $\ge_X$ be the corresponding value ordering.
We will define two different pairs $(\alpha, \beta)$ and
$(\alpha', \beta')$ in $\phi^*$.
Let $s'_\phi$ be $s_\phi$ restricted to $S_\phi - R_\phi$.
Let $\alpha$ be any outcome extending $u_\phi$ and $r_\phi$ and $s'_\phi$.
Define $\beta$ by:
$\beta(X)$ is an element other than $\alpha(X)$;
$\beta(Y) = s_\phi(Y)$ if $Y\in S_\phi$;
$\beta(Y) = \alpha(Y)$ for all other $Y$.
Then $(\alpha, \beta) \in \phi^*$,
and $\alpha$ and $\beta$ only differ on variable $X$ and variables $R_\phi \cap S_\phi$. The first variable in $\pi$ on which $\alpha$ and $\beta$ differ is $X$,
and thus, $\alpha(X) >_X \beta(X)$, since $\alpha \succceq_\pi \beta$.

%\commentphend
Now, define outcome $\alpha'$ which agrees with $\alpha$ except on $X$,
and outcome $\beta'$ which agrees with $\beta$ except on $X$,
and where $\alpha'(X) = \beta(X)$ and $\beta'(X) = \alpha(X)$.
By the same argument, we have $(\alpha', \beta') \in \phi^*$
and  $\alpha'(X) >_X \beta'(X)$, i.e.,
 $\beta(X) >_X \alpha(X)$, which is a contradiction, since $>_X$ is a total order.
\end{proof}

The following result characterises when
a lex model satisfies a non-strict preference statement in $\calLpq$.

\begin{proposition}
\label{pr:basic-pi-satisfies-phi-Lpq}
Let $\phi$ be a non-strict element of $\calLpq$ (so that $\phi = \phidestrict$),
which we write
as
$u_\phi r_\phi \ge u_\phi s_\phi $ \eqq~$T_\phi$.
Let $\pi\in\Glex$ be the model
$(Y_1, \ge_{1}), \ldots, (Y_k, \ge_{k})$.
Let $i$ be the smallest index such that
$Y_i \in R_\phi \cap S_\phi$,
or let $i = k+1$ if $R_\phi \cap S_\phi \cap V_\pi =\emptyset$.
Then, $\pi \models\phi$ if and only if
(i) $r_\phi(Y_i) >_i s_\phi(Y_i)$
if $R_\phi \cap S_\phi \cap V_\pi \not=\emptyset$,
and (ii) for all $j<i$,
\begin{itemize}
  \item[(a)] $Y_j\notin\WW_\phi$;
  \item[(b)] if $Y_j \in R_\phi \setminus S_\phi$
  then for all $y\in\und{Y_j}$,
  $r_\phi(Y_j) \ge_j y$, i.e., $r_\phi(Y_j)$ is the best value of $Y_j$; and
  \item[(c)] if $Y_j \in S_\phi \setminus R_\phi$  then for all $y\in\und{Y_j}$,
  $y \ge_j s_\phi(Y_j)$, i.e., $s_\phi(Y_j)$ is the worst value of $Y_j$.
\end{itemize}
In particular, if $V_\pi \subseteq T_\phi \cup U_\phi$ then $\pi\models\phi$.
\end{proposition}

%\commentph
This proposition can be seen to follow from the next proposition, which is an alternative version of the result.

%\commentph
\begin{proposition}
\label{pr:basic-pi-satisfies-phi-Lpq-version}
Let $\pi\in\Glex$
and $\phi$ be a non-strict element of $\calLpq$,
so that $\phi = \phidestrict$.
Let us say that $X\in V_\pi$ is \emph{definite} if
$X \in (R_\phi \cap S_\phi) \cup \WW_\phi$,
and that $X$ is \emph{relevant} if
$X \in R_\phi \cup S_\phi \cup \WW_\phi$
and there is no earlier definite variable in $V_\pi$.
Thus, the set of relevant variables consists of
the earliest definite variable (if there is one),
plus all earlier variables not in  $T_\phi$ or $U_\phi$.
As usual, we let $\ge_X$ be the total ordering
associated with $X$ in $\pi$.
Then, $\pi \models\phi$ if and only if
for all relevant variables $X$,
\begin{itemize}
  \item[(a)] $X\notin\WW_\phi$;
  \item[(b)] if $X \in R_\phi \cap S_\phi$ then
  $r_\phi(X) >_X s_\phi(X)$;
  \item[(c)] if $X \in R_\phi \setminus S_\phi$
  then for all $x\in\und{X}$,
  $r_\phi(X) \ge_X x$, i.e., $r_\phi(X)$ is the best value of $X$; and
  \item[(d)] if $X \in S_\phi \setminus R_\phi$  then for all $x\in\und{X}$,
  $x \ge_X s_\phi(X)$, i.e., $s_\phi(X)$ is the worst value of $X$.
\end{itemize}
In particular, if $V_\pi \subseteq T_\phi \cup U_\phi$ then $\pi\models\phi$.
\end{proposition}

\begin{proof}
%\commentphstart
First let us assume that $\pi \models\phi$, i.e., $\pi\models\phidestrict$,
and thus, $\succceq_\pi \ \supseteq\ \phi^*$.
We will prove that (a), (b), (c) and (d) hold.

%\commentph
We first define a pair $(\alpha_0, \beta_0)$ in $\phi^*$.
Let $s'_\phi$ be $s_\phi$ restricted to $S_\phi - R_\phi$.
Let $\alpha_0$ be any outcome extending $u_\phi$ and $r_\phi$ and $s'_\phi$.
Define $\beta_0$ by:
$\beta_0(Y) = s_\phi(Y)$ if $Y\in S_\phi$;
$\beta_0(Y) = \alpha_0(Y)$ for all other $Y$.
The only variables on which $\alpha_0$ and $\beta_0$ differ are
those in $R_\phi \cap S_\phi$, and we have $\alpha_0(R_\phi) = r_\phi$
and $\beta_0(S_\phi) = s_\phi$.
We then have $(\alpha_0, \beta_0) \in \phi^*$.

%\commentph
\noindent\ssk
(a): Suppose that  $X\in\WW_\phi$.
%FROM PROOF OF Lemma~\ref{le:phi-W-variable}.
Let $x$ be any element of $\und{X}$ other than $\alpha_0(X)$,
and let $x' = \alpha_0(X)$.
Define $\beta_1$ by $\beta(X) = x$,
and for all other $Y\in V\setminus\set{X}$,
$\beta_1(Y) = \beta_0(Y)$.
Also, define $\alpha_1$ by
$\alpha_1(X) = x$, and for all $Y\in V\setminus\set{X}$,
$\alpha_1(Y) = \alpha_0(Y)$.
It follows that
$(\alpha_0, \beta_1)$ and $(\alpha_1, \beta_0)$ are in $\phi^*$,
and thus, $\alpha_0 \succceq_\pi \beta_1$
and $\alpha_1 \succceq_\pi \beta_0$, because
$\succceq_\pi \ \supseteq\ \phi^*$.
The first variable in $\pi$ on which $\alpha_0$ and $\beta_1$ differ is $X$,
 and thus, $\alpha_0(X) >_X \beta_1(X)$, i.e., $x' >_X x$.
Similarly, the first variable in $\pi$ on which $\alpha_1$ and $\beta_0$ differ is $X$,
 and thus, $\alpha_1(X) >_X \beta_0(X)$, i.e., $x >_X x'$,
 contradicting the fact that $\ge_X$ is a total order.

%\commentph
\noindent\ssk
(b): Assume that $X \in R_\phi \cap S_\phi$,
and so $\alpha_0(X) \not= \beta_0(X)$.
Since $(\alpha_0, \beta_0) \in \phi^*$ we
have $\alpha_0 \succceq_\pi \beta_0$.
Let $Y$ be the first variable on which $\alpha_0$ and $\beta_0$
differ, so $Y\in R_\phi \cap S_\phi$.
$Y$ is thus a definite variable. % as is $X$.
Since $X$ is relevant, there is no earlier definite variable,
so $Y = X$, and  $X$ is the first variable on which $\alpha_0$ and $\beta_0$
differ.
$\alpha_0 \succceq_\pi \beta_0$ implies that
$\alpha_0(X) >_X \beta_0(X)$,
i.e., $r_\phi(X) >_X s_\phi(X)$, proving (b).

%\commentph
\noindent\ssk
(c):
Assume that $X \in R_\phi \setminus S_\phi$ and
  $\pi \models\phi$.
Choose any $x\in\und{X}$.
Let $\beta_2$ be an outcome that only differs with $\beta_0$
on $X$, and with $\beta_2(X) = x$.
Then, $(\alpha_0,\beta_2)\in\phi^*$, and
so $\alpha_0 \succceq_\pi \beta_2$, since $\pi \models\phi$.
Now, $\alpha_0$ and $\beta_2$ do not differ on any
earlier variables, since no earlier variable is in $R_\phi \cap S_\phi$,
because $X$ is relevant.
This implies that $\alpha_0(X) \ge_X \beta_2(X)$,
i.e., $r_\phi(X) \ge_X x$.

%\commentph
\noindent\ssk
(d):
Assume that $X \in S_\phi \setminus R_\phi$ and
  $\pi \models\phi$.
  The proof of (d) is analogous to that of (c).
Choose any $x\in\und{X}$.
Let $\alpha_2$ be an outcome that only differs with $\alpha_0$
on $X$, and with $\alpha_2(X) = x$.
Then, $(\alpha_2,\beta_0)\in\phi^*$, and so,
$\alpha_2(X) \ge_X \beta_0(X)$,
i.e., $x \ge_X s_\phi(X)$.

%\commentph
\ssk
To prove the converse, we now assume that for all relevant variables,
conditions (a), (b), (c) and (d) hold.
We will prove that $\pi\models\phi$.
It is sufficient to show that for all $(\alpha,\beta)\in\phi^*$
we have $\alpha \succceq_\pi \beta$.
So, consider any $(\alpha,\beta)\in\phi^*$.
If $\alpha(V_\pi) = \beta(V_\pi)$ then we have $\alpha \succceq_\pi \beta$,
so we can assume that $\alpha$ and $\beta$ differ on some variable in $V_\pi$;
let $X$ be the first such variable, and let
$\AA$ be the set of earlier variables, so that
$\alpha(\AA) = \beta(\AA)$.
By the definition of a lexicographic order,
to prove that $\alpha \succceq_\pi \beta$,
it is sufficient to prove that $\alpha(X) \ge_X \beta(X)$
(i.e., $\alpha(X) >_X \beta(X)$, since
$\alpha(X) \not=\beta(X)$).

%\commentph
Suppose that there exists a definite variable,
and let $Y$ be the earliest (according, as always to the $V_\pi$ ordering).
Then $Y$ is relevant.
By condition (a), $Y\notin\WW_\phi$ and so $Y\in R_\phi \cap S_\phi$,
 but then $\alpha(Y) = r_\phi(Y) \not= s_\phi(Y) = \beta(Y)$,
 so $\alpha(Y) \not=  \beta(Y)$.
 In particular this implies that $Y\notin\AA$,
 so $\AA$ contains no definite variable.
Since $\alpha(X) \not= \beta(X)$, we have
 $X\notin T_\phi \cup U_\phi$, so $X$ is relevant.

% \commentphend
 If $X\in R_\phi \cap S_\phi$ then the definition of $\phi^*$
 implies that $\alpha(X) = r_\phi(X)$ and $\beta(X) = s_\phi(X)$,
 and thus, $\alpha(X) >_X \beta(X)$, by condition (b).
If $X\in R_\phi \setminus S_\phi$ then $\alpha(X) = r_\phi(X)$
and condition (c) implies that $\alpha(X) \ge_X \beta(X)$.
%and thus, $\alpha(X) >_X \beta(X)$, since $\alpha(X) \not=\beta(X)$.
Similarly, if $X\in S_\phi \setminus R_\phi$
then condition (d) implies that $\alpha(X) \ge_X \beta(X)$.
\end{proof}

\noindent
The next result gives the extra conditions
required for satisfying strict statements.

\begin{proposition}
\label{pr:satisfaction-strict-statements}
Let $\phi\in\calLpq$ and $\pi\in\Glex$.
Then:
\begin{itemize}
  \item if $\phi$ is a fully strict statement then
  $\pi\models\phi$ if and only if $\pi \models \phidestrict$ and $R_\phi \cap S_\phi \cap V_\pi \not=\emp$;
  \item if $\phi$ is a weakly strict statement then
  $\pi\models\phi$ if and only if $\pi \models \phidestrict$ and
  $(R_\phi \cup S_\phi) \cap V_\pi \not=\emp$.
\end{itemize}
\end{proposition}

\begin{proof}
%\commentphstart
The definitions immediately imply that
if $\pi \models \phi$ then $\pi \models \phidestrict$,
so we can assume that in all cases
$\pi\models\phidestrict$.

%\commentph
Suppose that $\phi$ is a fully strict statement.
Now $\pi \models \phidestrict$ implies that
$\alpha\succceq_\pi \beta$ for all $(\alpha,\beta)\in\phi^*$.
Therefore,
$\pi\models\phi$ if and only if
for all $(\alpha,\beta)\in\phi^*$,
$\alpha\not\equiv_\pi \beta$, i.e., $\alpha(V_\pi) \not=\beta(V_\pi)$.
Lemma~\ref{le:equality-for-strict}(i) then implies that $\pi\models\phi$ if and only if
$R_\phi \cap S_\phi \cap V_\pi \not=\emp$.

%\commentphend
Assume now that $\phi$ is a weakly strict statement,
and also assume that
 $\pi \models \phidestrict$.
 We then have
$\pi\models\phi$ if and only if
there exists $(\alpha,\beta)\in\phi^*$ with
 $\alpha(V_\pi) \not=\beta(V_\pi)$,
 which, by Lemma~\ref{le:equality-for-strict}(ii),
 is if and only if
 $T_\phi \cup U_\phi \not\supseteq V_\pi$, i.e.,
 $(R_\phi \cup S_\phi \cup \WW_\phi) \cap V_\pi \not=\emp$.
 Now,
 %Proposition~\ref{pr:basic-pi-satisfies-phi-Lpq}
 Proposition~\ref{pr:basic-pi-satisfies-phi-Lpq-version}
 implies that if
 $\pi \models\phi$ and
 $\WW_\phi \cap V_\pi \not=\emp$ then $R_\phi \cap V_\pi \not=\emp$,
 and thus, $\pi\models\phi$ if and only if
 $(R_\phi \cup S_\phi) \cap V_\pi \not=\emp$.
\end{proof}

%\commentph
\begin{lemma}
\label{le:simult-decisive-R-S}
Suppose $\phi\in\calLpq$ is such that
$R_\phi = S_\phi$.
Then $\phi^*$ is simultaneously decisive
(see Definition~\ref{def:simult-decisive}).
\end{lemma}

\begin{proof}
%\commentphstartend
Let $\RR = \phi^*$ and
consider any $X\in V$ such that $\RR^{\downarrow X}$ is acyclic and  $\RR^{\downarrow X} \not\subseteq\ =$.
Proposition~\ref{pr:rs-local-ordering}
implies that $X\in R_\phi \cup S_\phi = R_\phi = S_\phi$,
and so $\RR^{\downarrow X}$ equals $\set{(r_\phi(X), s_\phi(X))}$,
which is irreflexive since $r_\phi(X) \not= s_\phi(X)$.
\end{proof}

%\commentph
As one would expect, both kinds of strict statements are
strict versions of $\phi^\RR$
(with a fully strict statement corresponding to $\psi_1$,
and a weakly strict statement corresponding to $\psi_2$,
in the discussion before Proposition~\ref{pr:phi-calR-strict-strongly-compos}).

%\commentph
\begin{lemma}
\label{le:strict-strict}
Suppose that $\phi\in\calLpq$ is either a fully strict statement or a weakly strict statement.
Then $\phi$ is a strict version of $\phi^\RR$,
where $\RR = \phi^*$.
\end{lemma}

Proposition~\ref{pr:strong-compos-alpha-beta-calR}
can be seen to imply that the non-strict
elements of the language $\calLpq$
are strongly compositional.
In fact this also holds for both kinds of strict statements
and certain negations.

\begin{theorem}
\label{th:phi-strongly-compos}
Consider any $\phi\in\calLpq$.
Then $\phi$ is strongly compositional
and
$\pi\modelsstar\phi$ if and only if $\pi\models\phidestrict$.
If $\phi$ is non-strict then
 $\neg\phi$ is compositional;
 if also, $R_\phi = S_\phi$ then
 $\neg\phi$ is strongly compositional,
 and $\pi\modelsstar\neg\phi$ if and only if
 either $\pi\models\neg\phi$ or $V_\pi \cap S_\phi = \emp$.
\end{theorem}

\begin{proof}
%\commentphstart
Let $\RR = \phi^*$.
First suppose that $\phi$ is either a fully strict statement or a weakly strict statement.
For all $\pi\in\Glex$ we have
$\pi\models\phidestrict$ $\iff$ $\pi\models\phi^\RR$.
Then, by Lemma~\ref{le:strict-strict}, $\phi$ is a strict version of $\phi^\RR$.
Proposition~\ref{pr:phi-calR-strict-strongly-compos} implies that
$\phi$ is strongly compositional and, for $\pi\in\Glex$,
$\pi\modelsstar\phi$ if and only if $\pi\models\phi^\RR$,
which is if and only if $\pi\models\phidestrict$.

%\commentph
Now suppose that $\phi$ is non-strict.
Then for all $\pi\in\Glex$ we have
$\pi\models\phi$ $\iff$ $\pi\models\phi^\RR$,
and thus also,
$\pi\modelsstar\phi$ $\iff$ $\pi\modelsstar\phi^\RR$.
Proposition~\ref{pr:phi-calR-strongly-compos} implies that
$\phi^\RR$ is strongly compositional,
and $\pi\modelsstar\phi^\RR$ $\iff$
$\pi\models\phi^\RR$ for any $\pi\in\Glex$.
Thus, $\phi$ is strongly compositional,
and, for all $\pi\in\Glex$ we have $\pi\modelsstar\phi$ $\iff$
 $\pi\models\phidestrict$, since $\phidestrict = \phi$.
Proposition~\ref{pr:neg-phi-calR-compos} implies that
$\neg\phi$ is compositional.

%\commentph
Now, assume also that $R_\phi = S_\phi$.
Then  $\RR = \phi^*$ is simultaneously decisive,
by Lemma~\ref{le:simult-decisive-R-S}.
Proposition~\ref{pr:neg-strongly-compos} implies that
$\neg\phi^\RR$ and thus $\neg\phi$ is strongly compositional
and for all $\pi\in\Glex$, $\pi\modelsstar\neg\phi$ $\iff$ $\pi\in\calM_{\neg\phi^\RR}$.

%\commentphend
We have:
$\pi\in\calM_{\neg\phi^\RR}$ $\iff$
either (i) $\pi\not\models\phi^\RR$ or
(ii) for every variable $X\in V_\pi$,
either $\RR^{\downarrow X}$ is not acyclic or $\RR^{\downarrow X} \subseteq\ =$.
(ii) holds if and only if for every $X\in V_\pi$,
we have
$X\notin R_\phi \cup S_\phi$
(i.e., $X\notin S_\phi$, since $R_\phi = S_\phi$), so
(ii) holds if and only if $V_\pi \cap S_\phi = \emp$.
Thus, $\pi\modelsstar\neg\phi$ holds if and only if
either $\pi\models\neg\phi$ or $V_\pi \cap S_\phi = \emp$.
\end{proof}

\subsection{Checking Consistency for Subsets of $\calLpqn$}
\label{subsec:checking-cons-calLpqn}

Theorem~\ref{th:phi-strongly-compos} suggests the feasibility of checking consistency of subsets of the language $\calLpqn$,
which is $\calLpq$ with certain negated statements also included.
Formally,
define $\calLpqn$ to be the union of $\calLpq$
with
$\set{\neg\phi \st \phi\in\calLpq,
\phi \textit{ non-strict, and } R_\phi = S_\phi}$.

We use the method of Section~\ref{subsec:checking-cons}
to determine the consistency of a set of preference statements $\Gamma\subseteq\calLpqn$,
by incrementally extending a  maximal $\modelsstar$-model $\pi$ of $\Gamma$, and then checking whether or not $\pi\models\Gamma$ holds;
this makes use of Propositions~\ref{pr:basic-pi-satisfies-phi-Lpq} and~\ref{pr:satisfaction-strict-statements}.

\newcommand\bestelt{\textrm{Best}}
\newcommand\worstelt{\textrm{Worst}}
\newcommand\pospairs{\textrm{Pos}}
\newcommand\negpairs{\textrm{Neg}}
\newcommand\pairs{\textrm{Pairs}}

\newcommand{\Gammabar}{\overline{\Gamma}}

Let $\Gamma\subseteq\calLpqn$,
let $X\in V$ and let $\pi\in\Glex$.
We make the following definitions,
where  $\Gammabar$ is the set of
all $\phi\in\Gamma\cap\calLpq$ such that
$R_\phi \cap S_\phi \cap V_\pi = \emp$.
\begin{itemize}
  \item $\bestelt_\Gamma^\pi(X) =
  \set{r_\phi(X) \st \phi\in\Gammabar \ \& \
  X\in R_\phi\setminus S_\phi}$.
  \item  $\worstelt_\Gamma^\pi(X) =
  \set{s_\phi(X) \st \phi\in\Gammabar \ \& \
   X\in S_\phi\setminus R_\phi}$.
   \item $\pairs_\Gamma^\pi(X) = \pospairs_\Gamma^\pi(X) \cup \negpairs_\Gamma^\pi(X)$,
       where
\end{itemize}
    $\pospairs_\Gamma^\pi(X) =
  \set{(r_\phi(X), s_\phi(X)) : \phi\in\Gammabar \ \& \
  X\in R_\phi\cap S_\phi}$; and
   $\negpairs_\Gamma^\pi(X)$ is the set of  all
   pairs
  $(s_\phi(X), r_\phi(X))$ such that $\neg\phi\in\Gamma$ and
  $T_\phi \cup U_\phi \supseteq V_\pi$,
    and $X\in R_\phi $($= S_\phi$).

%\commentph
\begin{lemma}
\label{le:Bestelt-etc-properties}
Suppose that $\Gamma\subseteq\calLpqn$.
Let $X \in V$ and let $\ge_X$ be a total ordering on $\und{X}$,
and let $\pi' = \pi \circ (X,\ge_X)$.
Suppose that
$\pi' \modelsstar \Gamma$. Then the following hold:
\begin{itemize}
  \item For all  $ x \in \bestelt_\Gamma^\pi(X)$
and $x' \in \und{X}$ we have
$x \ge_X x'$. In particular then,
$|\bestelt_\Gamma^\pi(X)| \le 1$.
  \item  For all  $ x \in \worstelt_\Gamma^\pi(X)$
and $x' \in \und{X}$ we have
$x' \ge_X x$. In particular then,
$|\worstelt_\Gamma^\pi(X)| \le 1$.
\item If
$(x,x') \in \pairs_\Gamma^\pi(X)$
then $x \ge_X x'$.
\end{itemize}
\end{lemma}

\begin{proof}
%\commentphstart
Since $\pi' \modelsstar \Gamma$ we have,
by Theorem~\ref{th:phi-strongly-compos},
$\pi'\models\phidestrict$ for $\phi\in\Gamma\cap\calLpq$,
and for $\neg\phi \in \Gamma$,
either $\pi'\models\neg\phi$ or $V_{\pi'} \cap S_\phi = \emp$.
Recall the definitions of definite and relevant variables in Proposition~\ref{pr:basic-pi-satisfies-phi-Lpq-version}.
Given $\phi\in\Gamma\cap\calLpq$,
we have that
if $\phi$ is such that $R_\phi \cap S_\phi \cap V_\pi = \emp$ and $X\in R_\phi\cup S_\phi$
then $X$ is relevant given $\phidestrict$ and $\pi'$.
This is because $X$ would only not be relevant if
there were an earlier definite variable $Y$ in $\pi'$
and thus in $V_\pi$;
we'd then have
$Y\notin\WW_\phi$, by
Proposition~\ref{pr:basic-pi-satisfies-phi-Lpq-version}
and so $Y\in R_\phi \cap S_\phi$, which contradicts
$R_\phi \cap S_\phi \cap V_\pi = \emp$.

%\commentph
Suppose that $x \in \bestelt_\Gamma^\pi(X)$.
By definition, there exists
$\phi\in\Gamma\cap\calLpq $ such that
$x = r_\phi(X)$ and
  $R_\phi \cap S_\phi \cap V_\pi = \emp$ and $X\in R_\phi\setminus S_\phi$.
Since $X$ is relevant given $\phidestrict$ and $\pi'$,
Proposition~\ref{pr:basic-pi-satisfies-phi-Lpq-version}
implies that
for all $x'\in\und{X}$,
  $x \ge_X x'$.
Since $\ge_X$ is a total order, there can be at most one element $x$
 in $\bestelt_\Gamma^\pi(X)$.
A similar argument shows that if  $x \in \worstelt_\Gamma^\pi(X)$
then for all $x'\in \und{X}$,
we have $x' \ge_X x$ which implies that
$|\worstelt_\Gamma^\pi(X)| \le 1$.

%\commentph
Suppose that $(x,x') \in \pospairs_\Gamma^\pi(X)$.
Then, by definition, there exists $\phi\in\Gamma\cap\calLpq$
such that $r_\phi(X) = x$ and $ s_\phi(X) = x'$
and $R_\phi \cap S_\phi \cap V_\pi = \emp$ and $X\in R_\phi\cap S_\phi$.
Since $X$ is relevant given $\phidestrict$ and $\pi'$,
Proposition~\ref{pr:basic-pi-satisfies-phi-Lpq-version}
implies that
$x \ge_X x'$.

%\commentphend
Suppose that $(x,x') \in \negpairs_\Gamma^\pi(X)$.
Then there exists $\neg\phi\in\Gamma$ with
$s_\phi(X)= x$ and  $r_\phi(X) = x'$ and
  $T_\phi \cup U_\phi \supseteq V_\pi$ and  $X\in R_\phi= S_\phi$.
Since $V_{\pi'} \cap S_\phi \not= \emp$
and $\pi'\modelsstar\neg\phi$,
we have $\pi'\models\neg\phi$,
i.e., $\pi'\not\models\phi$.
The condition $T_\phi \cup U_\phi \supseteq V_\pi$,
 using
Proposition~\ref{pr:basic-pi-satisfies-phi-Lpq-version}
implies that $\pi\models\phidestrict$, i.e.,
$\pi\models\phi$, since $\phi$ is non-strict.
Also, $X$ is relevant given $\phidestrict$ and $\pi'$,
so, $\pi'\not\models\phi$ implies,
using
Proposition~\ref{pr:basic-pi-satisfies-phi-Lpq-version},
that
$r_\phi(X) \not\ge_X s_\phi(X))$,
and thus, $s_\phi(X) >_X r_\phi(X)$
and $x \ge_X x'$.
\end{proof}

Given $\Gamma$ and $\pi\in\Glex$ with
$\pi \modelsstar\Gamma$,
we say that $X$ \emph{can be chosen next} if:
$X\in V - V_\pi$ and
\begin{itemize}
  \item if $\phi\in\Gamma\cap\calLpq$
  and $R_\phi \cap S_\phi \cap V_\pi = \emp$ then
  $X \notin \WW_\phi$;
  \item $\pairs_\Gamma^\pi(X)$ is acyclic;
  \item $|\bestelt_\Gamma^\pi(X)| \le 1$ and $|\worstelt_\Gamma^\pi(X)| \le 1$;
  \item if $x \in \bestelt_\Gamma^\pi(X)$ then
  $x$ is undominated in $\pairs_\Gamma^\pi(X)$,
  i.e., there exists no element of the form $(x', x)$ in $\pairs_\Gamma^\pi(X)$;
  \item if $x \in \worstelt_\Gamma^\pi(X)$ then
  $x$ is not dominating in $\pairs_\Gamma^\pi(X)$,
  i.e., there exists no element of the form $(x, x')$ in $\pairs_\Gamma^\pi(X)$.
\end{itemize}
We also say that
 $(X,\ge_X)$ is a \emph{valid extension of $\pi$} if
 (i) $X$ can be chosen next; and
 (ii) $\ge_X$ $\supseteq$ $\pairs_\Gamma^\pi(X)$;
  (iii) if $x\in \bestelt_\Gamma^\pi(X)$ then
 $x$ is the best element in $\und{X}$
 with respect to $\ge_X$ (so that $x\ge y$
 for all $y\in\und{X}$);
 and (iv) if $x'\in \worstelt_\Gamma^\pi(X)$ then
 $x'$ is the worst element in $\und{X}$
 with respect to $\ge_X$.

%\commentph
We state a basic lemma that follows easily
from
Proposition~\ref{pr:basic-pi-satisfies-phi-Lpq-version}.

%\commentph
\begin{lemma}
\label{le:W-not-in-models}
Suppose that $\phi\in \calLpq$ and
 $\pi\in\Glex$ are such that
 $\pi \models \phi$ and
 $V_\pi \cap R_\phi \cap S_\phi = \emp$.
 Then, $V_\pi \cap \WW_\phi = \emp$.
\end{lemma}

The following result states the  conditions needed for minimally extending $\pi$ to
maintain the $\modelsstar$-satisfaction of $\Gamma$.

\begin{proposition}
\label{pr:Valid-extensions}
Suppose that $\Gamma\subseteq\calLpqn$,
and that $\pi \modelsstar\Gamma$.
Let $X$ be a variable in $V-V_\pi$
and let $\pi' = \pi \circ (X,\ge_X)$,
where $\ge_X$ is a total ordering on $\und{X}$.
Then
$\pi' \modelsstar \Gamma$
if and only if $(X,\ge_X)$ is a valid extension of $\pi$.
\end{proposition}

\begin{proof}
%\commentphstart
Since $\pi \modelsstar\Gamma$,
Theorem~\ref{th:phi-strongly-compos} implies that
$\pi\models\phidestrict$ for $\phi\in\Gamma\cap\calLpq$,
and for $\neg\phi \in \Gamma$,
either $\pi\models\neg\phi$ or $V_{\pi} \cap S_\phi = \emp$
(since $\neg\phi \in \calLpqn$ implies that
$\phi\in\calLpq$ and $\phi$ is non-strict, and $R_\phi = S_\phi$).

%\commentph
\ssk\noindent
$\Leftarrow:$
We will first prove that if $(X,\ge_X)$ is a valid extension of $\pi$
then $\pi' \modelsstar \Gamma$.
For $\phi\in\Gamma\cap\calLpq$,
we have $\pi'\modelsstar\phi$ if and only if $\pi'\models\phidestrict$,
by Theorem~\ref{th:phi-strongly-compos}.
Also, for $\neg\phi \in \Gamma$ we have
$\pi'\modelsstar\neg\phi$ if and only if
 either $\pi'\models\neg\phi$ or $V_{\pi'} \cap S_\phi = \emp$.

%\commentph
Consider any $\phi\in\Gamma\cap\calLpq$.
Since  $\pi\models\phidestrict$, it follows using  Lemma~\ref{le:alpha-pi-beta-Gamma-sequential}
that $\pi'\models\phidestrict$
if and only if  $\ge_X$ $\supseteq$ $\phi^{\downarrow X}_{V_\pi}$.

%\commentph
Consider any $(x,x') \in \phi^{\downarrow X}_{V_\pi}$.
We need to show that $(x,x') \in\ \ge_X$,
i.e., that $x \ge_X x'$;
this will then imply that
$\ge_X$ $\supseteq$ $\phi^{\downarrow X}_{V_\pi}$,
and hence, $\pi'\models\phidestrict$.
Since $\phi^{\downarrow X}_{V_\pi}$ is non-empty, we have,
using Proposition~\ref{pr:rs-local-ordering},
that $R_\phi \cap S_\phi \cap V_\pi = \emp$.
This implies that $X \notin \WW_\phi$,
because  $(X,\ge_X)$ is a valid extension of $\pi$.
If $x = x'$ then clearly, $x \ge_X x'$,
This covers the cases when $X \in T_\phi$ and $X \in U_\phi$
(see Proposition~\ref{pr:rs-local-ordering}).

%\commentph
If $X \in R_\phi \cap S_\phi$, then, by  Proposition~\ref{pr:rs-local-ordering},
$\phi^{\downarrow X}_{V_\pi} = \set{(r_\phi(X), s_\phi(X))}.$
Thus,  $(x,x')\in \pospairs_\Gamma^\pi(X)$,
so $x \ge_X x'$.

%\commentph
If $X \in R_\phi \setminus S_\phi$
then $r_\phi(X) \in \bestelt_\Gamma^\pi(X)$.
Proposition~\ref{pr:rs-local-ordering} implies that
$x = r_\phi(X)$, and thus,
$x \ge_X x'$.
Similarly, if $X \in S_\phi \setminus R_\phi$
then $s_\phi(X) \in \worstelt_\Gamma^\pi(X)$.
Proposition~\ref{pr:rs-local-ordering} implies that
$x' = s_\phi(X)$, and thus,
$x \ge_X x'$.
This completes the proof that,
for any $\phi\in\Gamma\cap\calLpq$, we have
$\ge_X$ $\supseteq$ $\phi^{\downarrow X}_{V_\pi}$,
and hence, $\pi'\models\phidestrict$,
and thus, $\pi'\modelsstar\phi$.

%\commentph
Now suppose that $\neg\phi\in\Gamma$,
and so $\phi$ is non-strict and $R_\phi = S_\phi$.
We will show that $\pi' \modelsstar \neg\phi$.
Since
$\pi\modelsstar \neg\phi$ we have either $\pi\models\neg\phi$ or $V_{\pi} \cap S_\phi = \emp$.
If  $\pi\models\neg\phi$,
and so $\succceq_\pi \not\supseteq \phi^*$, then the fact that $\pi'$ extends $\pi$ implies that $\succceq_{\pi'} \subseteq \succceq_\pi$
(e.g., using Lemma~\ref{le:extends-pi-basic}),
and thus, $\succceq_{\pi'}\not\supseteq \phi^*$,
and therefore $\pi'\models\neg\phi$ and $\pi'\modelsstar\neg\phi$.
We now thus have only to consider the case when $\pi\models\phi$
and $V_{\pi} \cap S_\phi = \emp$.
%The fact that $\pi\models\phi$ and
% $V_\pi \cap S_\phi (=R_\phi) = \emp$
This implies,
 using Lemma~\ref{le:W-not-in-models},
 that $V_\pi \cap \WW_\phi = \emp$,
 and thus, $V_\pi \subseteq T_\phi \cup U_\phi$.
If $X\notin S_\phi$ then $V_{\pi'} \cap S_\phi = \emp$,
and so, $\pi' \modelsstar \neg\phi$.
Now assume that $X\in S_\phi$.
This implies that $(s_\phi(X), r_\phi(X)) \in \negpairs_\Gamma^\pi(X)$.
Because $(X,\ge_X)$ is a valid extension of $\pi$,
we have
$s_\phi(X) \ge_X r_\phi(X)$, i.e.,
 $s_\phi(X) >_X r_\phi(X)$, since $s_\phi(X) \not= r_\phi(X)$.
It cannot be the case that
$\pi' \models \phi$, since then we would have
$\pi'\models\phidestrict$ and thus,
using Lemma~\ref{le:alpha-pi-beta-Gamma-sequential},
$\ge_X$ $\supseteq$ $\phi^{\downarrow X}_{V_\pi}$,
which implies
$r_\phi(X) \ge_X s_\phi(X)$ using Proposition~\ref{pr:rs-local-ordering},
contradicting $s_\phi(X) >_X r_\phi(X)$.
We therefore have
 $\pi' \models \neg\phi$,
and thus, $\pi' \modelsstar \neg\phi$.

%\commentph
\ssk\noindent
$\Rightarrow$:
Assume now that $\pi' \modelsstar \Gamma$;
we will show that  $(X,\ge_X)$ is a valid extension of $\pi$.
We have
$X\in V - V_\pi$.
Since $\pi' \modelsstar \Gamma$ we have
$\pi'\models\phidestrict$ for $\phi\in\Gamma\cap\calLpq$,
and for $\neg\phi \in \Gamma$,
either $\pi'\models\neg\phi$ or $V_{\pi'} \cap S_\phi = \emp$.
For $\phi\in\Gamma\cap\calLpq$ we then have
$\ge_X$ $\supseteq$ $\phi^{\downarrow X}_{V_\pi}$,
by Lemma~\ref{le:alpha-pi-beta-Gamma-sequential}.
Lemma~\ref{le:W-not-in-models} implies that
if $\phi\in\Gamma\cap\calLpq$
  and $R_\phi \cap S_\phi \cap V_\pi = \emp$ then
  $X \notin \WW_\phi$.
Lemma~\ref{le:Bestelt-etc-properties} implies that
there is at most one element in $\bestelt_\Gamma^\pi(X)$
and at most one element in $\worstelt_\Gamma^\pi(X)$.
Also if  $(x,x') \in \pairs_\Gamma^\pi(X)$
then $x \ge_X x'$ and so $\ge_X$ extends $\pairs_\Gamma^\pi(X)$,
and thus, $\pairs_\Gamma^\pi(X)$ is acyclic.
The same lemma also implies that
if $x \in \bestelt_\Gamma^\pi(X)$
then for all $x' \in \und{X}$,
$x \ge_X x'$,
and thus, by the acyclicity of $\ge_X$,
$x$ is undominated in $\pairs_\Gamma^\pi(X)$.
A similar argument shows that
if  $x \in \worstelt_\Gamma^\pi(X)$
then $x$ is not dominating in $\pairs_\Gamma^\pi(X)$.
This completes the proof that
$(X,\ge_X)$ is a valid extension of $\pi$.
\end{proof}

In summary, at each stage
of the incremental algorithm
 we see if there is a variable $X$ that can be chosen next,
and if so, generate a valid extension;
if not, we then have generated a maximal $\modelsstar$-model
$\pi$ of $\Gamma$.
We check consistency of $\Gamma$ by determining if
$\pi$ satisfies $\Gamma$.
Using the fact that
$|\pairs_\Gamma^\pi(X)| \le |\Gamma|$,
it can be shown that the overall complexity of checking
that $\Gamma$ ($\subseteq\calLpqn$) is consistent is
$\calO(|V|^2 |\Gamma|)$.

%\commentph
\subsubsection{Checking $\Gamma$-Satisfaction of Maximal $\modelsstar$-Model}

%\commentph
This section considers the following problem:
Suppose that $\pi \modelsstar \Gamma$,
where $\Gamma \subseteq\calLpqn$.
Under what conditions do we have
$\pi \models \Gamma$?

%\commentph
\begin{proposition}
\label{pr:final-satisfaction-check}
Suppose that $\Gamma \subseteq\calLpqn$
and that $\pi \modelsstar \Gamma$.
\begin{itemize}
  \item If $\phi \in\Gamma \cap \calLpq$ and $\phi$ is non-strict, then $\pi\models\phi$.
  \item If $\phi \in\Gamma \cap \calLpq$ and $\phi$ is fully strict, then $\pi\models\phi$ $\iff$ $R_\phi \cap S_\phi \cap V_\pi \not=\emp$.
  \item If $\phi \in\Gamma \cap \calLpq$ and $\phi$ is weakly strict, then $\pi\models\phi$ $\iff$ $(R_\phi \cup S_\phi) \cap V_\pi \not=\emp$.
    \item If $\neg\phi\in\Gamma$, where $\phi$ is a non-strict element of $\calLpq$ with $R_\phi = S_\phi$,
  we have $\pi\models \neg\phi$ $\iff$ $V_\pi \not\subseteq T_\phi \cup U_\phi$.
\end{itemize}
Thus,
$\pi \models \Gamma$ if and only if
\begin{itemize}
  \item for all fully strict statements $\phi$
in $\Gamma \cap \calLpq$, $R_\phi \cap S_\phi \cap V_\pi \not=\emp$;
  \item for all weakly strict statements $\phi$
in $\Gamma \cap \calLpq$, $(R_\phi \cup S_\phi) \cap V_\pi \not=\emp$;
    \item for all $\neg\phi\in\Gamma$, where $\phi$ is a non-strict element of $\calLpq$ with $R_\phi = S_\phi$,
  we have $V_\pi \not\subseteq T_\phi \cup U_\phi$.
\end{itemize}
\end{proposition}

\begin{proof}
%\commentphstart
First consider any $\phi\in\Gamma\cap\calLpq$.
We have that $\pi\modelsstar \phi$,
which implies, by Theorem~\ref{th:phi-strongly-compos},
that $\pi \models \phidestrict$.
Thus, if $\phi$ is non-strict then $\pi\models\phi$,
showing the first bullet point.
We can then use Proposition~\ref{pr:satisfaction-strict-statements}
to imply the second and third bullet points.

%\commentph
Now consider an element of the form $\neg\phi$  in $\Gamma$.
Theorem~\ref{th:phi-strongly-compos} implies that
either $\pi \models\neg\phi$ or $V_\pi \cap S_\phi = \emp$.
Thus, if $\pi\models\phi$ then $V_\pi \cap S_\phi = \emp$,
and so $V_\pi \cap \WW_\phi = \emp$,
by Lemma~\ref{le:W-not-in-models},
and so $V_\pi \cap(R_\phi \cup S_\phi \cup \WW_\phi) = \emp$,
i.e., $V_\pi \subseteq T_\phi \cup U_\phi$.
Conversely, if $V_\pi \subseteq T_\phi \cup U_\phi$
then it follows using
Proposition~\ref{pr:basic-pi-satisfies-phi-Lpq-version}
that $\succceq_\pi$ $\supseteq$ $\phi^*$,
and hence $\pi \models \phidestrict$,
and so, $\pi\models\phi$, since $\phi$ is a non-strict statement.
This proves the fourth bullet point.

%\commentphend
The second half of the result follows from the first half.
\end{proof}

\section{Optimality with Respect to Sets of Compositional Preference Statements}
\label{sec:optimality}

We consider a finite set of alternatives $\calA \subseteq\und{V}$,
and we assume that we have elicited a set $\Gamma$ of preference statements;
we would like to find the optimal alternatives among $\calA$.
As we will see, there are several natural definitions of \emph{optimal} \cite{GelainPRVW10,WilsonOMahonyAICS2011}.
We compare some of these in this section and analyse their computational cost in the context of lex models and compositional statements in the next section.

We define
pre-order relation $\succceq_\Gamma$ on outcomes by
$\alpha \succceq_\Gamma \beta$ $\iff$
$\Gamma \modelslex \alpha\ge\beta$,
and we define equivalence relation $\equiv_\Gamma$ by
$\alpha \equiv_\Gamma \beta$ $\iff$
$\Gamma \modelslex \alpha\equiv\beta$,
i.e., if $\alpha$ and $\beta$ are equivalent in all models of $\Gamma$.
Using a similar argument to that for Theorem~\ref{th:max-modelsstar-consistent},
it can be shown that,
for compositional $\Gamma$, we have  $\Gamma \modelslex \alpha\equiv\beta$ holds
if and only if $\alpha \equiv_\pi \beta$,
 where $\pi$ is an arbitrary maximal model of $\Gamma$.
 We also define $\succ_\Gamma$ to be the strict part of $\succceq_\Gamma$, so that $\alpha \succ_\Gamma \beta$
if and only if $\alpha \succceq_\Gamma \beta$ and
$\alpha \not\equiv_\Gamma \beta$.
We then say that $\alpha$ \emph{strictly dominates} $\beta$.

%%\commentph
%For compositional $\Gamma$,
%and arbitrary model $\pi$ of $\Gamma$,
%Proposition~\ref{pr:compos-max-models} implies that
%$\Gamma \modelslex \alpha\equiv\beta$
%if and only if $\alpha$ and $\beta$ are equivalent in $\pi$.

We define $\CSD_\Gamma(\calA)$ (`Can Strictly Dominate') to be the set of maximal, i.e., undominated, elements of $\calA$ w.r.t. $\succ_\Gamma$.
$\alpha\in\CSD_\Gamma(\calA)$  if and only if
for all $\beta\in\calA$ which are not
%$\Gamma$-equivalent (i.e.,
$\equiv_\Gamma$-equivalent to $\alpha$
there exists some $\pi\models\Gamma$ with $\alpha\succ_\pi\beta$.

We define $\OO_\pi(\calA)$ to be the subset of the alternatives that are optimal in model $\pi$, i.e.,
$\set{\alpha\in\calA \st \forall\beta\in\calA, \alpha \succceq_\pi\beta}$.
We say that $\alpha\in\calA$ is \emph{necessarily optimal in $\calA$},
written $\alpha\in\NO_\Gamma(\calA)$,
if $\alpha$ is optimal in every model,
i.e., if for all $\pi\in\Glex$ with $\pi\models\Gamma$
we have $\alpha\in\OO_\pi(\calA)$.
%\commentph
This holds if and only if for all $\beta\in\calA$,
and for all $\pi$ with $\pi\models\Gamma$ we have
$\alpha \succceq_\Gamma \beta$.

We say that $\alpha$ is \emph{possibly optimal},
written $\alpha\in\PO_\Gamma(\calA)$,
if $\alpha$ is optimal in some model of $\Gamma$,
so that
$\PO_\Gamma(\calA) = \bigcup_{\pi\models\Gamma} \OO_\pi(A)$.
Similarly, we say that
$\alpha\in\POM_\Gamma(\calA)$ if
$\alpha$ is optimal in some maximal model of $\Gamma$.
%\commentph
Thus we have
$\POM_\Gamma(\calA) = \bigcup_{\pi\modelsmax\Gamma} \OO_\pi(A)$,
where $\pi\modelsmax\Gamma$ means that
$\pi$ is a maximal model of $\Gamma$.
$\alpha$ is \emph{possibly strictly optimal} in $\calA$,
%(given $\Gamma$),
written $\alpha\in\PSO_\Gamma(\calA)$,
if there exists
some $\pi\models\Gamma$ with $\OO_\pi(A)\ni \alpha$
and $\Gamma\models\alpha\equiv\beta$ for all $\beta\in\OO_\pi(A)$.
Thus $\alpha$ is in $\PSO_\Gamma(\calA)$ if
there is a model of $\Gamma$ in which $\alpha$ is optimal,
and all other optimal elements are equivalent to $\alpha$.

%\commentph
Given $\Gamma\subseteq\calL$, we say that $\alpha$  is \emph{strictly optimal} (within $\calA$)  with respect to $\pi$
if $\alpha$ is optimal in $\pi$ and any other optimal element is equivalent to $\alpha$,
i.e., $\alpha\in\OO_\pi(\calA)$ and
$\Gamma\models\alpha\equiv\beta$
for all $\beta\in\OO_\pi(\calA)$.
We write
$\SO^\Gamma_\pi(\calA)$ for the set of such elements.

%\commentph
For $B\subseteq\und{V}$, we say that $B$ are all $\Gamma$-equivalent if for all $\alpha,\beta\in B$,
we have
$\Gamma\models\alpha\equiv\beta$.
If $\OO_\pi(\calA)$ are all $\Gamma$-equivalent then
$\SO^\Gamma_\pi(\calA) = \OO_\pi(\calA)$,  otherwise,
$\SO^\Gamma_\pi(\calA) = \emp$.
We always have that $\SO^\Gamma_\pi(\calA)$ are all $\Gamma$-equivalent.

Let $\Opt_\Gamma^\calA(\alpha)$ be the set of models $\pi$ of $\Gamma$
that make $\alpha$ optimal in $\calA$,
i.e., $\set{\pi\models\Gamma\st \OO_\pi(\calA)\ni\alpha}$.
We define $\alpha\in\MPO_\Gamma(\calA)$ if $\Opt_\Gamma^\calA(\alpha)$
is maximal, in the sense that there exists no $\beta\in\calA$
with $\Opt_\Gamma^\calA(\beta)$ a strict superset of $\Opt_\Gamma^\calA(\alpha)$.
We say that $\alpha\in\MPO_\Gamma(\calA)$ is \emph{maximally possibly optimal} in $\calA$ given $\Gamma$;
this holds if and only if there is no alternative that is optimal in the same set of lex models and more.

%\commentph
For $\alpha\in\calA$,
let $\Delta^\calA_\alpha = \set{\alpha\ge\beta \st \beta\in\calA}$.

%\commentph
We give some basic properties.

%\commentph
\begin{lemma}
\label{le:OO-PO-basic}
Let $\Gamma\subseteq\calL$ be a set of preference statements,
and let $\pi$ and $\pi'$ be lex models,
and let $\alpha$ be an element of set of alternatives $\calA\subseteq\und{V}$. Then the following all hold.
\begin{itemize}
  \item If $\pi'$ extends $\pi$ then $\OO_{\pi'}(\calA) \subseteq \OO_\pi(\calA)$.
  \item $\alpha\in\OO_\pi(\calA)$ $\iff$ $\pi\models\Delta^\calA_\alpha$
  \item $\pi \models \Gamma\cup\Delta^\calA_\alpha$
        $\iff$ $\pi\in\Opt_\Gamma^\calA(\alpha)$
  \item $\alpha\in\PO_\Gamma(\calA)$ $\iff$ $\Gamma\cup\Delta^\calA_\alpha$ is consistent
        $\iff$ $\Opt_\Gamma^\calA(\alpha)$ is non-empty.
  \item $\Gamma\cup\Delta^\calA_\alpha \models\alpha\equiv\beta$ $\iff$  $\Opt_\Gamma^\calA(\beta) \supseteq  \Opt_\Gamma^\calA(\alpha)$.
\end{itemize}
\end{lemma}

\begin{proof}
%\commentphstart
Assume that $\pi'$ extends $\pi$.
Consider any $\alpha \in \OO_{\pi'}(\calA)$,
so that, for all $\beta\in\calA$,
$\alpha \succceq_{\pi'}\beta$.
By Lemma~\ref{le:extends-pi-basic},
if $\alpha \succceq_{\pi'} \beta$ then
$\alpha \succceq_{\pi} \beta$.
Therefore, for all $\beta\in\calA$,
$\alpha \succceq_{\pi}\beta$,
and so $\alpha \in \OO_\pi(\calA)$.

%\commentph
We have: $\pi\models\Delta^\calA_\alpha$ if and only if for all $\beta\in\calA$,
 $\alpha \succceq_\pi \beta$, which is if and only if
 $\alpha \in \OO_\pi(\calA)$.

%\commentph
 $\pi \models \Gamma\cup\Delta^\calA_\alpha$
 if and only if $\pi\models\Gamma$
 and $\alpha \in \OO_\pi(\calA)$,
 which is if and only if
        $\pi\in\Opt_\Gamma^\calA(\alpha)$.

%\commentph
 $\alpha\in\PO_\Gamma(\calA)$
 if and only if
 there exists some $\pi$ with $\pi\models\Gamma$
 and $\alpha \in \OO_\pi(\calA)$.
 By the second bullet point,
 this holds if and only if there exists $\pi$ with
 $\pi \models \Gamma\cup\Delta^\calA_\alpha$,
 i.e.,  $\Gamma\cup\Delta^\calA_\alpha$ is consistent.
 This is also equivalent to
 $\Opt_\Gamma^\calA(\alpha)$ being non-empty by the third bullet point.

%\commentph
First suppose that
$\Gamma\cup\Delta^\calA_\alpha \models\alpha\equiv\beta$,
and consider any $\pi\in\Opt_\Gamma^\calA(\alpha)$.
Then, $\pi \models \Gamma\cup\Delta^\calA_\alpha$,
and thus, $\pi\models\alpha\equiv\beta$,
and so, $\alpha \equiv_\pi \beta$.
This implies that
$\beta \in \OO_\pi(\calA)$ and hence, $\pi \in \Opt_\Gamma^\calA(\beta)$.

%\commentphend
Conversely, suppose that
$\Opt_\Gamma^\calA(\beta) \supseteq \Opt_\Gamma^\calA(\alpha)$,
and consider any $\pi$ such that
$\pi \models \Gamma\cup\Delta^\calA_\alpha$;
we then have $\alpha \succceq_\pi \beta$.
Then, by the third bullet point,
$\pi\in \Opt_\Gamma^\calA(\alpha)$,
and so, $\pi\in \Opt_\Gamma^\calA(\beta)$,
which implies that
$\pi \models \Gamma\cup\Delta^\calA_\beta$.
This entails that $\beta \succceq_\pi \alpha$,
and thus  $\alpha \equiv_\pi \beta$,
i.e., $\pi \models \alpha\equiv\beta$.
We have shown that
$\Gamma\cup\Delta^\calA_\alpha \models\alpha\equiv\beta$.
\end{proof}

\noindent
Without making assumptions about $\Gamma$ we have the following properties, which follow from basic arguments,
that apply in a very general context \cite{WilsonOMahonyAICS2011} (for proofs see also \cite{OMahonyThesis13}).

\begin{proposition}
\label{pr:MODS-results}
Consider  any $\calA\subseteq\und{V}$ and $\Gamma\subseteq\calL$.
Then, the following all hold.
(i) $\NO_\Gamma(\calA) \cup \PSO_\Gamma(\calA)
\subseteq \MPO_\Gamma(\calA) \cap \EXT_\Gamma(\calA)$;
(ii) $\EXT_\Gamma(\calA) \subseteq \CSD_\Gamma(\calA) \cap \PO_\Gamma(\calA)$;
(iii) $\MPO_\Gamma(\calA) \subseteq \PO_\Gamma(\calA)$;
(iv) $\MPO_\Gamma(\calA) \cap \EXT_\Gamma(\calA)$ is always non-empty.
(v) If $\NO_\Gamma(\calA)$ is non-empty then
$\NO_\Gamma(\calA) = \MPO_\Gamma(\calA) = \EXT_\Gamma(\calA) = \CSD_\Gamma(\calA)$.
\end{proposition}

We can visualise these relations in the following diagram, where $A \rightarrow B$ represents the relation $A\subseteq B$.
\tikzstyle{line} = [draw, -latex']
    \resizebox{\columnwidth}{!}
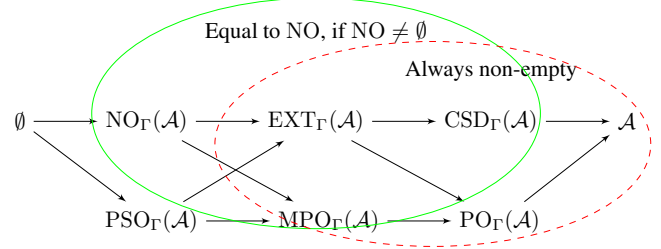
{
\begin{tikzpicture}[node distance = 2cm, auto]
    % Place nodes
     \node (empty) {$\emptyset$};
     \node [right = 1cm of empty] (NO) {$\NO_\Gamma(\calA)$};
     \node [below right = 1.4cm of empty] (PSO) {$\PSO_\Gamma(\calA)$};
     \node [right = 1cm of NO] (EXT) {$\EXT_\Gamma(\calA)$};
     \node [right = 1cm of PSO] (MPO) {$\MPO_\Gamma(\calA)$};
     \node [right = 1cm of EXT] (CSD) {$\CSD_\Gamma(\calA)$};
     \node [right = 1cm of MPO] (PO) {$\PO_\Gamma(\calA)$};
     \node [right = 1cm of CSD] (A) {$\calA$};
     
     \node [above = -.3cm of EXT] (eqToNO) {};
      \node [above = 0.8cm of EXT] {Equal to $\NO$, if $\NO \neq \emptyset$};
     \node [below left= -.1cm of CSD] (nonempt) {};
     \node [above = .2cm of CSD] {Always non-empty};
     
\path [line] (empty) -- (NO);
\path [line] (empty) -- (PSO);
\path [line] (NO) -- (EXT);
\path [line] (NO) -- (MPO);
\path [line] (PSO) -- (EXT);
\path [line] (PSO) -- (MPO);
\path [line] (EXT) -- (CSD);
\path [line] (EXT) -- (PO);
\path [line] (MPO) -- (PO);
\path [line] (CSD) -- (A);
\path [line] (PO) -- (A);

\draw[green] (eqToNO) ellipse (3.5cm and 1.8cm);
\draw[red, dashed] (nonempt) ellipse (3.4cm and 1.6cm);
\end{tikzpicture}
}

The following lemmas and propositions will extend these results by relations of $\POM_\Gamma(\calA)$ and the case where $\Gamma$ is compositional.

%\commentph
\begin{lemma}
\label{le:POM-equals-PSO}
For any $\calA\subseteq\und{V}$ and $\Gamma\subseteq\calL$ we have $\PSO_\Gamma(\calA) \subseteq \POM_\Gamma(\calA)$.
If $\Gamma$ is compositional then
$\PSO_\Gamma(\calA) = \POM_\Gamma(\calA)$.
\end{lemma}

\begin{proof}
%\commentphstart
Suppose that $\alpha\in\PSO_\Gamma(\calA)$,
so there exists $\pi\in\Glex$ with
$\pi\models\Gamma$, and $\OO_\pi(\calA)\ni\alpha$,
and $\Gamma\models\alpha\equiv\beta$ for all $\beta\in\OO_\pi(\calA)$.
Let $\pi'$ be any maximal model of $\Gamma$ that extends $\pi$.
Choose some $\beta$ that is optimal in $\pi'$,
i.e., $\beta\in\OO_{\pi'}(\calA)$.
Then, using Lemma~\ref{le:OO-PO-basic}, $\beta\in\OO_{\pi}(\calA)$,
and thus, $\Gamma\models\alpha\equiv\beta$, so $\pi'\models\alpha\equiv\beta$,
which implies that $\alpha\in\OO_{\pi'}(\calA)$,
and thus, $\alpha\in\POM_\Gamma(\calA)$.

%\commentphend
Assume now that $\Gamma$ is compositional.
Let $\alpha\in\POM_\Gamma(\calA)$,
so there exists $\pi\in\Glex$ with $\pi\modelsmax\Gamma$
and $\alpha\in\OO_\pi(\calA)$.
Proposition~\ref{pr:compos-max-models} implies that
for all $\beta\in\OO_\pi(\calA)$ we have $\Gamma\models\alpha\equiv\beta$,
and thus, $\alpha\in\PSO_\Gamma(\calA)$.
\end{proof}

%\commentph
\begin{lemma}
\label{le:POM-supseteq-MPO}
For any $\calA\subseteq\und{V}$ and compositional $\Gamma\subseteq\calL$,
%if $\Gamma$ is compositional then
we have
$\MPO_\Gamma(\calA) \subseteq \POM_\Gamma(\calA)$.
\end{lemma}

\begin{proof}
%\commentphstartend
We will prove that
$(\PO_\Gamma(\calA) - \POM_\Gamma(\calA)) \cap \MPO_\Gamma(\calA) = \emp$.
Since, $\MPO_\Gamma(\calA) \subseteq \PO_\Gamma(\calA)$, this implies
that $\MPO_\Gamma(\calA) \subseteq \POM_\Gamma(\calA)$.
Let $\alpha\in \PO_\Gamma(\calA) - \POM_\Gamma(\calA)$.
By Lemma~\ref{le:OO-PO-basic},
$\Gamma\cup\Delta^\calA_\alpha$ is consistent;
we choose some maximal model $\pi$ of $\Gamma\cup\Delta^\calA_\alpha$.
In particular,  $\OO_\pi(\calA) \ni\alpha$.
Choose some maximal model $\pi'$ of $\Gamma$ extending $\pi$.
Then, $\alpha\notin \OO_{\pi'}(\calA)$, since $\alpha\notin\POM_\Gamma(\calA)$.
Choose some $\beta\in\OO_{\pi'}(\calA)$, and thus, $\beta\in\OO_{\pi}(\calA)$,
by Lemma~\ref{le:OO-PO-basic}. This implies $\alpha\equiv_\pi\beta$,
and thus, by Proposition~\ref{pr:compos-max-models},
$\Gamma\cup\Delta^\calA_\alpha \models\alpha\equiv\beta$.
This implies that $\Opt_\Gamma^\calA(\beta) \supseteq  \Opt_\Gamma^\calA(\alpha)$,
by Lemma~\ref{le:OO-PO-basic}.
Since $\pi' \in \Opt_\Gamma^\calA(\beta) - \Opt_\Gamma^\calA(\alpha)$,
we have that $\alpha\notin\MPO_\Gamma(\calA)$.
\end{proof}

%\commentph
A straight-forward argument implies that
$\PSO_\Gamma(\calA) \subseteq \MPO_\Gamma(\calA)$.
Lemmas~\ref{le:POM-equals-PSO} and~\ref{le:POM-supseteq-MPO} then imply the following.

%\commentph
\begin{proposition}
\label{pr:POM-equals-PSO-equals-MPO}
For any $\calA\subseteq\und{V}$ and $\Gamma\subseteq\calL$ we have $\PSO_\Gamma(\calA) \subseteq \POM_\Gamma(\calA)
\cap \MPO_\Gamma(\calA)$.
If $\Gamma$ is compositional then
$\PSO_\Gamma(\calA) = \POM_\Gamma(\calA) = \MPO_\Gamma(\calA)$.
\end{proposition}

Let $\pi_1, \ldots, \pi_k$ be a finite sequence of models.
Define $\calA_{\pi_1}$ to be
$\OO_{\pi_1}(\calA)$.
For $i=1,\ldots, k$ we iteratively define $\calA_{\pi_1, \ldots, \pi_i}$
to be
$\OO_{\pi_i}(\calA_{\pi_1, \ldots, \pi_{i-1}})$.
%$\max_{\pi_i}(\calA_{\pi_1, \ldots, \pi_{i-1}})$.
We define the \emph{extreme elements} $\EXT_\Gamma(\calA)$ as follows.
$\alpha\in\EXT_\Gamma(\calA)$ if and only if there exists
a sequence $\pi_1, \ldots, \pi_k$ of models of $\Gamma$
such that $\calA_{\pi_1, \ldots, \pi_k}\ni\alpha$
and for all $\beta\in\calA_{\pi_1, \ldots, \pi_k}$, $\Gamma\models\alpha\equiv\beta$.
Therefore, $\alpha\in\EXT_\Gamma(\calA)$ if there is a sequence of lex models such that iteratively maximising with respect to each lex model in turn leads to a set containing $\alpha$ and only other alternatives that are $\Gamma$-equivalent to $\alpha$.

%%\commentph
%\begin{lemma}
%\label{le:EXT-supseteq-PSO}
%For any $\calA\subseteq\und{V}$ and  $\Gamma\subseteq\calL$,
%we have
%$\EXT_\Gamma(\calA) \supseteq \PSO_\Gamma(\calA)$.
%\end{lemma}
%
%\begin{proof}
%%\commentphstartend
%Suppose $\alpha\in\PSO_\Gamma(\calA)$,
%so there exists $\pi\in\Glex$ with
%$\pi\models\Gamma$, and $\OO_\pi(\calA)\ni\alpha$,
%and $\Gamma\models\alpha\equiv\beta$ for all $\beta\in\OO_\pi(\calA)$.
%Using the singleton sequence $\pi$, we have $\alpha\in\EXT_\Gamma(\calA)$.
%\end{proof}

%\commentph
\begin{lemma}
\label{le:EXT-circ-basic}
Let $\pi_1, \ldots, \pi_k$ be a finite sequence of models,
and let $\pi = \pi_1\circ\cdots\circ\pi_k$.
Then, $\calA_{\pi_1, \ldots, \pi_k} = \calA_\pi = \OO_{\pi}(\calA)$.
\end{lemma}

\begin{proof}
%\commentphstart
We first show that for arbitrary $\pi,\pi'\in\Glex$,
$\OO_{\pi'}(\calA_\pi) = \calA_{\pi\circ\pi'}$,
i.e., $\OO_{\pi'}(\OO_\pi(\calA)) = \OO_{\pi\circ\pi'}(\calA)$.

%\commentph
Consider an element $\alpha\in \OO_{\pi'}(\OO_\pi(\calA))$;
we will show that $\alpha\succceq_{\pi\circ\pi'} \beta$ for every $\beta\in\calA$,
showing that $\alpha\in\OO_{\pi\circ\pi'}(\calA)$.
We have that $\alpha\in\OO_\pi(\calA)$, which implies
$\alpha\succceq_\pi \beta$.
If $\alpha\succ_\pi \beta$ then $\alpha\succ_{\pi\circ\pi'} \beta$,
by Lemma~\ref{le:extends-pi-basic}.
Otherwise,  we have  $\alpha\equiv_\pi \beta$, which implies that $\beta\in\OO_\pi(\calA)$,
and thus, $\alpha\succceq_{\pi'}\beta$.
 Lemma~\ref{le:composition-basic} implies that $\alpha\succceq_{\pi\circ\pi'}\beta$.

%\commentph
Conversely, assume that
$\alpha\in\OO_{\pi\circ\pi'}(\calA)$.
Consider any $\beta\in\OO_\pi(\calA)$.
We need to show that $\alpha\succceq_{\pi'}\beta$.
We have $\alpha\succceq_{\pi\circ\pi'}\beta$.
Lemma~\ref{le:extends-pi-basic} implies that $\alpha\succceq_{\pi}\beta$,
which implies that $\alpha\in\OO_\pi(\calA)$, and also $\alpha\equiv_\pi \beta$.
 Since $\alpha\succceq_{\pi\circ\pi'}\beta$,
 Lemma~\ref{le:composition-basic} implies  $\alpha\succceq_{\pi'}\beta$, as required.

%\commentphend
We now prove the result by induction.
It is trivial for $k=1$.
%We have just proved that it holds for $k=2$.
Now, $\calA_{\pi_1, \ldots, \pi_k} = \OO_{\pi_k}(\calA_{\pi_1, \ldots, \pi_{k-1}})$,
which by the inductive hypothesis equals
$\OO_{\pi_k}(\calA_{\pi_1\circ\cdots\circ\pi_{k-1}})$,
which equals $\calA_{\pi_1\circ\cdots\circ\pi_{k}}$, by the argument above.
\end{proof}

%\commentphnot
%\begin{proposition}
%\label{pr:POM-equals-PSO-equals-MPO-equals-EXT}
%For any $\calA\subseteq\und{V}$ and $\Gamma\subseteq\calL$ we have $\PSO_\Gamma(\calA) \subseteq \POM_\Gamma(\calA)
%\cap \MPO_\Gamma(\calA) \cap \EXT_\Gamma(\calA)$.
%If $\Gamma$ is compositional then
%$\PSO_\Gamma(\calA) = \POM_\Gamma(\calA) = \MPO_\Gamma(\calA)
%= \EXT_\Gamma(\calA)$.
%\end{proposition}

%\commentph
The optimality class $\EXT_\Gamma(\calA)$ turns out also to be equivalent to $\PSO_\Gamma(\calA)$ when
$\Gamma$ is compositional.

%\commentph
\begin{proposition}
\label{pr:EXT-equals-PSO-compos}
Consider  any $\calA\subseteq\und{V}$ and compositional $\Gamma\subseteq\calL$.
Then $\EXT_\Gamma(\calA) = \PSO_\Gamma(\calA)$.
\end{proposition}

\begin{proof}
%\commentphstartend
Proposition~\ref{pr:MODS-results} implies $\EXT_\Gamma(\calA) \supseteq \PSO_\Gamma(\calA)$.
To prove the converse, suppose that $\alpha\in\EXT_\Gamma(\calA)$.
Then there exists a sequence $\pi_1, \ldots, \pi_k$ of models of $\Gamma$
such that $\calA_{\pi_1, \ldots, \pi_k}\ni\alpha$
and for all $\beta\in\calA_{\pi_1, \ldots, \pi_k}$, $\Gamma\models\alpha\equiv\beta$.
By Lemma~\ref{le:EXT-circ-basic},
$\alpha\in\OO_{\pi}(\calA)$, where
 $\pi = \pi_1\circ\cdots\circ\pi_k$,
 and $\Gamma\models\alpha\equiv\beta$ for all $\beta\in\OO_{\pi}(\calA)$.
Since $\Gamma$ is compositional, $\pi\models\Gamma$,
and thus, $\alpha\in\PSO_\Gamma(\calA)$.
\end{proof}

%\commentph
\noindent
Propositions~\ref{pr:MODS-results},~\ref{pr:EXT-equals-PSO-compos} and~\ref{pr:POM-equals-PSO-equals-MPO} imply the following result, showing that there are substantial simplifications of the optimality classes when $\Gamma$ is compositional.

\begin{theorem}
\label{th:optimality-classes-compositional}
Consider  any $\calA\subseteq\und{V}$ and compositional $\Gamma\subseteq\calL$.
Then
$\NO_\Gamma(\calA) \subseteq \PSO_\Gamma(\calA) = \EXT_\Gamma(\calA) = \MPO_\Gamma(\calA)
= \POM_\Gamma(\calA) \subseteq \CSD_\Gamma(\calA)
\cap \PO_\Gamma(\calA)$.
\end{theorem}

\begin{proof}
%\commentphstartend
Proposition~\ref{pr:POM-equals-PSO-equals-MPO} implies that
$\PSO_\Gamma(\calA) =  \MPO_\Gamma(\calA) = \POM_\Gamma(\calA)$.
Proposition~\ref{pr:EXT-equals-PSO-compos} implies that
$\EXT_\Gamma(\calA) = \PSO_\Gamma(\calA)$.
Proposition~\ref{pr:MODS-results} implies that
$\NO_\Gamma(\calA) \subseteq  \EXT_\Gamma(\calA) \subseteq \CSD_\Gamma(\calA) \cap \PO_\Gamma(\calA)$, completing the proof.
\end{proof}

We can summarise the relations in the following diagram.
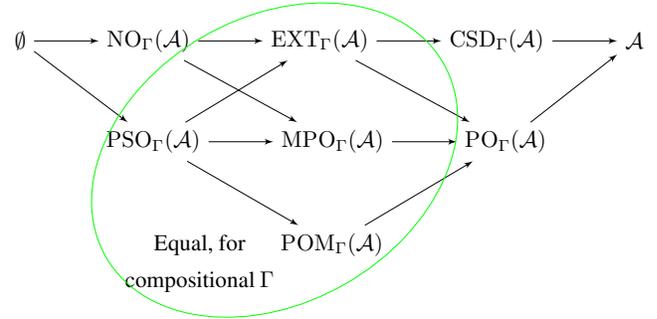
\tikzstyle{line} = [draw, -latex']
    \resizebox{\columnwidth}{!}{
\begin{tikzpicture}[node distance = 2cm, auto]
    % Place nodes
     \node (empty) {$\emptyset$};
     \node [right = 1cm of empty] (NO) {$\NO_\Gamma(\calA)$};
     \node [below right = 1.4cm of empty] (PSO) {$\PSO_\Gamma(\calA)$};
     \node [right = 1cm of NO] (EXT) {$\EXT_\Gamma(\calA)$};
     \node [right = 1cm of PSO] (MPO) {$\MPO_\Gamma(\calA)$};
     \node [below right = 1.4cm of PSO] (POM) {$\POM_\Gamma(\calA)$};
     \node [right = 1cm of EXT] (CSD) {$\CSD_\Gamma(\calA)$};
     \node [right = 1cm of MPO] (PO) {$\PO_\Gamma(\calA)$};
     \node [right = 1cm of CSD] (A) {$\calA$};
     
%     \node [below right = 0.3cm of PSO] (eqComp) {};
     \node [below left = -.2cm of MPO] (eqComp) {};
      \node [left = 0.3cm of POM] (ww){Equal, for};
      \node [below = 0cm of ww] {compositional $\Gamma$};
%     \node [below left= -.1cm of CSD] (nonempt) {};
%     \node [above = .2cm of CSD] {Always non-empty};
     
\path [line] (empty) -- (NO);
\path [line] (empty) -- (PSO);
\path [line] (NO) -- (EXT);
\path [line] (NO) -- (MPO);
\path [line] (PSO) -- (EXT);
\path [line] (PSO) -- (MPO);
\path [line] (EXT) -- (CSD);
\path [line] (EXT) -- (PO);
\path [line] (MPO) -- (PO);
\path [line] (CSD) -- (A);
\path [line] (PO) -- (A);
\path [line] (PSO) -- (POM);
\path [line] (POM) -- (PO);

\draw[green, rotate = 30] (eqComp) ellipse (3cm and 2.2cm);
%\draw[red, dashed] (nonempt) ellipse (3.4cm and 1.6cm);
\end{tikzpicture}
}

\section{Computing Optimal Solutions}
\label{sec: experiments}

In this section, we analyse the efficiency of computing $\PO_\Gamma$, $\PSO_\Gamma$, $\CSD_\Gamma$ and $\NO_\Gamma$ theoretically and experimentally
for $\Gamma\subseteq\calLpqn$.
Note, that by Theorem~\ref{th:optimality-classes-compositional} since $\Gamma\subseteq\calLpqn$ is compositional, $\EXT_\Gamma = \MPO_\Gamma = \POM_\Gamma = \PSO_\Gamma$ and thus $\PSO_\Gamma$ can be chosen to represent all of these classes.

\ssk\noindent
\textbf{Theoretical Running Times:}
One approach to compute $\op(\calA)$ for $\op \in \{\PO_\Gamma,\PSO_\Gamma,\CSD_\Gamma,\NO_\Gamma\}$ is to test membership for every alternative separately,
i.e., checking if $\alpha \in \op(\calA)$ for all $\alpha \in \calA$.
By the results in Section~\ref{subsec:checking-cons-calLpqn}, we can decide consistency for $g$ statements and $n$ variables in $\calO(n^2 g)$.
To test if $\alpha \in \PO_\Gamma(\calA)$, we test whether $\Gamma \cup  \{\alpha \geq \beta \mid \beta \in \calA- \{\alpha\}\}$ is consistent in $\calO(n^2 (g + m))$, where $|\calA| = m$.
Similarly, we test if $\alpha \in \PSO_\Gamma(\calA)$ in $\calO(n^2 (g + m))$, by checking if $\Gamma \cup  \{\alpha > \beta \mid \beta \in \calA,  \beta \not\equiv_\Gamma \alpha\}$ is consistent.
To test if $\alpha \in \CSD_\Gamma(\calA)$, we check for all $\beta \in \calA$ with $\beta \not\equiv_\Gamma \alpha$ if $\Gamma \cup  \{\alpha > \beta\}$ is consistent in $\calO(m n^2 g)$.
To test if $\alpha \in \NO_\Gamma(\calA)$, we test for all $\beta \in \calA- \{\alpha\}$ if $\Gamma \cup  \{\alpha < \beta\}$ is consistent in $\calO(m n^2 g)$.
The theoretical running times for computing $\op(\calA)$ for $\op \in \{\PO_\Gamma,\PSO_\Gamma,\CSD_\Gamma,\NO_\Gamma\}$ are summarised by the following table.
\resizebox{\columnwidth}{!}{
\begin{tabular}{l|c|c|c}
					& \shortspecialcell{$\PO_\Gamma, \PSO_\Gamma$}		& \shortspecialcell{$\CSD_\Gamma$}			& \shortspecialcell{$\NO_\Gamma$}		\\ \hline
\shortspecialcell{ Running time}	& $\calO(m n^2 (g + m))$		& \shortspecialcell{$\calO(m^2 n^2 g)$}	& \shortspecialcell{$\calO(m^2 n^2 g)$} 	
\end{tabular}
}
\label{tab: RT gap}

Let $M_\op(m)$ denote the worst case running time of testing if an outcome $\alpha \in \calA$ is in $\op(\calA)$ for $|\calA| = m$ and operator $\op$.
Then the running time to compute $\op(\calA)$ can be estimated by $O(m M_\op(m))$.

The work in~\cite{WilsonRM15} describes the algorithm ``IncrementalO'' to compute $\op(\calA)$ for optimality operators $\op$ and alternatives $\calA = \{\alpha_1, \dots, \alpha_m\}$ in an incremental way by testing if $\alpha_i \in \op(\{\alpha_1, \dots, \alpha_{i-1}\})$ and if so computing $\op(\{\alpha_1, \dots, \alpha_{i}\}$.
This algorithm may be used to compute $\PO_\Gamma,\PSO_\Gamma$ and $\CSD_\Gamma$ since all of these operators satisfy path independence and thus are optimality operators.

Under the assumption that $M_\op(m)$ is monotonically increasing in $m$,
we can estimate the running time $M_\op(m)$ of the algorithm described in~\cite{WilsonRM15} to compute $\op(m)$ by $O(\sum_{i=1, \dots, m} i M_\op(i))$ in the worst case and $O(m M_\op(1))$ in the best case.

The algorithm ``IncrementalO'' can not be applied to compute $\NO_\Gamma$, since $\NO_\Gamma$ does not satisfy path independence. However, a similar incremental computation is possible. 
Furthermore, by preprocessing the set of alternatives and including only one representative of every equivalence class w.r.t. $\Gamma$-equivalence (in $O(n^2g + m n)$), we get that $|\NO_\Gamma(\calA')| \leq 1$.

The incremental computation of computing $\op(\calA)$ for $\op \in \{\PO_\Gamma,\PSO_\Gamma,\CSD_\Gamma,\NO_\Gamma\}$ results in the following theoretical best and worst case running times.
\resizebox{\columnwidth}{!}{
\begin{tabular}{l|c|c|c}
					& \shortspecialcell{$\PO_\Gamma, \PSO_\Gamma$}		& \shortspecialcell{$\CSD_\Gamma$}			& \shortspecialcell{$\NO_\Gamma$}		\\ \hline
\shortspecialcell{ Best case}	& $\calO(m n^2 g)$		& \shortspecialcell{$\calO(m n^2 g)$}	& \shortspecialcell{$\calO(m n^2 g)$} 	\\ \hline
\shortspecialcell{ Worst case}	& $\calO(m^2 n^2 (g + m))$		& \shortspecialcell{$\calO(m^3 n^2 g)$}	& \shortspecialcell{$\calO(m^2 n^2 g)$} 	
\end{tabular}
}
\label{tab: RT gap1}

However, our experiments indicate no or only little improvement for the running times of the incremental computation compared with the previously described computation of separate membership tests.
We will thus concentrate on the running time results for the computation of separate membership tests.
%
%\ssk\noindent
%\textbf{Instances:}
% We randomly generated 10 consistent instances each for $g \in \{100, \dots, 1000\}$ preference statements $\Gamma \subseteq \calLpqn$ (equally divided in fully strict, weakly strict, non-strict and negated non-strict) and $n \in \{100, \dots, 1000\}$ variables with random domain sizes 2 or 3.
%Here, the values of assignments involved in $\Gamma$ were chosen randomly so that $\pi_S \vDash \Gamma$, where $\pi_S= (X_1, \geq_1)\dots(X_n,\geq_n)$ with trivial strict value orders.
%Here, the values of outcomes involved in $\Gamma$ were chosen randomly so that $\pi_S \vDash \Gamma$, where $\pi_S$ is the trivial complete lex model.
%In this way we ensure that all variables are involved in a maximal model.

\ssk\noindent
\textbf{Experimental Results:}
We implemented the methods for computing $\op(\calA)$ for $\op \in \{\PO_\Gamma,\PSO_\Gamma,\CSD_\Gamma,\NO_\Gamma\}$ in Java Version 1.8 and conducted our experiments independently on a 2.66Ghz quad-core processor with 12GB memory.
We randomly generated 10 consistent instances each, with $g \in \{100, \dots, 1000\}$ preference statements $\Gamma \subseteq \calLpqn$ (equally divided into fully strict, weakly strict, non-strict and negated non-strict) and $n \in \{100, \dots, 1000\}$ variables with random domain sizes 2 or 3.
Here, to ensure that $\Gamma$ is consistent,
the values of assignments involved in $\Gamma$ were chosen randomly so that $\pi_S \models \Gamma$,
for a given fixed lex order
 $\pi_S= (X_1, \geq_1)\dots(X_n,\geq_n)$.
% with trivial strict value orders.
Since the experiments showed similar behavior for $m=100, \dots, 500$ alternatives% (as expected by the theoretical results running times increase with increasing $m$)
, we fix $m=100$ in the following analysis.
In every instance, the experiments result in $\NO_\Gamma(\calA) \subseteq \PO_\Gamma(\calA) = \PSO_\Gamma(\calA) \subseteq \CSD_\Gamma(\calA)$ which extends the relations from
Theorem~\ref{th:optimality-classes-compositional}.
%Section~\ref{sec:optimality}.
%** NO cardinality now hidden.
%and $|\NO_\Gamma(\calA)| \leq 1$.
%
Here, $\NO_\Gamma(\calA)$ is non-empty more often the higher the number of preference statements $g$ and the lower the number of variables $n$ is.
The cardinality of sets $\PO_\Gamma(\calA), \PSO_\Gamma(\calA), \CSD_\Gamma(\calA)$ tends to be higher the lower the number of preference statements $g$ and the higher the number of variables $n$ is.
The diagram below shows mean cardinalities (over 10 instances each) of $\PO_\Gamma(\calA)=\PSO_\Gamma(\calA)$, $\CSD_\Gamma(\calA)$, $\NO_\Gamma(\calA)$ for different $n$ with $g = 1000$.

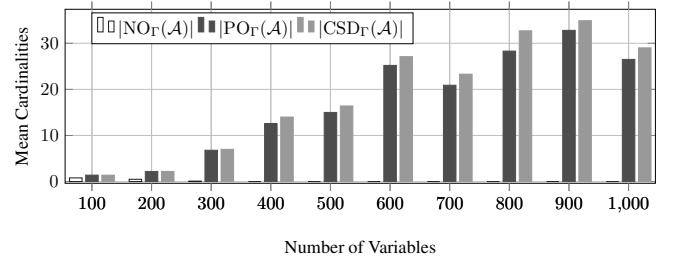
\resizebox{\columnwidth}{!}{
\begin{tikzpicture}
\begin{axis}[
  	ybar,
	bar width= 21,
 	height=4.95cm,
  	width=13cm,
  	ymajorgrids = true,
  	xmajorgrids = true,
  	enlarge x limits=0.05,
  	enlarge y limits=0.01,
  	xlabel={Number of Variables},
 	 ylabel={Mean Cardinalities},
  	extra x ticks = {100, 200, 300, 400, 500, 600, 700, 800, 900, 1000,1100},
  	ymin = 0,
  	ymax = 37,
  	xmin = 100,
  	xmax = 1000,
  	legend style={at={(0.59,0.8)},anchor=south east,legend columns=-1},
  	cycle list = {black!90,black!70,black!40}
   ]
\addplot+[] coordinates{
(100, 0.8)(200, 0.5)(300, 0.1)(400, 0)(500, 0)(600, 0)(700, 0)(800, 0)(900, 0)(1000, 0)(1100, 0)
};
\addplot+[fill,text=black] coordinates{
(100, 1.4)(200, 2.2)(300, 6.8)(400, 12.6)(500, 15)(600, 25.2)(700, 20.9)(800, 28.3)(900, 32.8)(1000, 26.5)(1100, 30)
};
\addplot+[fill,text=black] coordinates{
(100, 1.4)(200, 2.2)(300, 7.0)(400, 14.0)(500, 16.4)(600, 27.1)(700, 23.3)(800, 32.7)(900, 34.9)(1000, 29)(1100, 0)
};
\legend{$|\NO_\Gamma(\calA)|$,$|\PO_\Gamma(\calA)|$, $|\CSD_\Gamma(\calA)|$}
\end{axis}
\end{tikzpicture}
}

The mean running times (over 10 instances) for computing $\PO_\Gamma(\calA)$ and $\PSO_\Gamma(\calA)$ are similar for each instance class.
They range from 0.11s to
%16.47s,
16.5s,
and tend to increase with $n$ and $g$. % (only 6\% outliers).
For computing $\CSD_\Gamma(\calA)$ they range from 2.69s to
%659.03s
659s
and tend to increase with $n$ % (only 4\% outliers)
but show no clear trend by scaling $g$.
For computing $\NO_\Gamma(\calA)$ they range from 0.09s to
%18.81s
18.8s
and tend to increase with $n$ and $g$.
%_\Gamma(\calA)
Thus, in all instances $\PO$, $\PSO$, $\CSD$ and $\NO$ are computed efficiently.
In 71\% of the instance classes $\PO$ is computed faster than $\NO$.
Although $\CSD$ can be computed efficiently too, it is clearly computationally more expensive than $\PO$, $\PSO$ and $\NO$.

%\commentph
The following figure shows the mean running times (over 10 instances) of computing $\PO_\Gamma(\calA)$, $\CSD_\Gamma(\calA)$ and $\NO_\Gamma(\calA)$ for different $n$ with $g = 1000$.
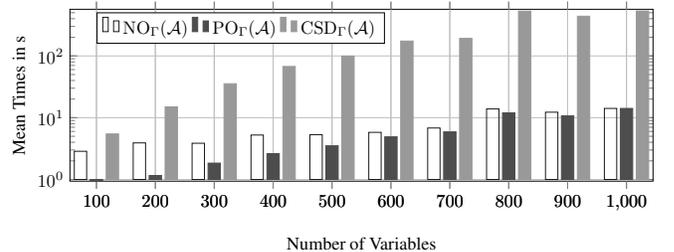
\resizebox{\columnwidth}{!}{
\begin{tikzpicture}
%\node[label={180:{\small50+}}] at (.2,3.1) {};
\begin{axis}[
  	ybar,
  	ymode=log,
	bar width= 21,
 	height=4.95cm,
  	width=13cm,
  	ymajorgrids = true,
  	xmajorgrids = true,
  	enlarge x limits=0.05,
  	enlarge y limits=0.01,
  	xlabel={Number of Variables},
 	 ylabel={Mean Times in s},
  	extra x ticks = {100, 200, 300, 400, 500, 600, 700, 800, 900, 1000,1100},
  	ymin = 1,
  	xmin = 100,
  	xmax = 1000,
  	legend style={at={(0.54,0.8)},anchor=south east,legend columns=-1},
  	cycle list = {black!90,black!70,black!40}
   ]
\addplot+[] coordinates{
(100,2.86)
(200,3.91)
(300,3.86)
(400,5.24)
(500, 5.31)
(600, 5.79)
(700, 6.81)
(800, 13.86)
(900, 12.27)
(1000, 14.13)
};
\addplot+[fill,text=black] coordinates{
(100,0.59)
(200,1.16)
(300,1.84)
(400,2.62)
(500, 3.51)
(600, 4.90)
(700, 5.90)
(800, 12.00)
(900, 10.75)
(1000, 14.13)
};
\addplot+[fill,text=black] coordinates{
(100,5.46)
(200,14.97)
(300,35.47)
(400,67.48)
(500, 99.53)
(600, 173.98)
(700, 192.23)
(800, 531.57)
(900, 439.30)
(1000, 533.94)
};
\node[label={360:{70+}}] at (-10,110) {};
\legend{$\NO_\Gamma(\calA)$,$\PO_\Gamma(\calA)$, $\CSD_\Gamma(\calA)$}
\end{axis}
\end{tikzpicture}
}

\ssk\noindent
\textbf{Evaluation:}
In practice, it is often desired to compute relatively small sets of optimal solutions, e.g., not to overwhelm a user with too many choices.
%But since $\NO_\Gamma$ is often empty,
Since $\NO_\Gamma$ is usually empty,
%(especially for high numbers of variables and low numbers of statements),
computing $\PSO_\Gamma$ or $\PO_\Gamma$
may be the best choice.
%is the most advantageous choice.
This is further supported by the analysis of computational cost.
In the experiments, computing
$\PSO_\Gamma$ and $\PO_\Gamma$ scales much better than %$\NO_\Gamma$ and
$\CSD_\Gamma$,
with the number of statements and variables,
%, in theory,
and is often faster than $\NO_\Gamma$.
% in the experiments.

\section{Discussion}
\label{sec:discussion}

We have developed a new approach for preference inference based on lexicographic models, using a notion of strong compositionality that allows a greedy algorithm.
This is shown to be polynomial for relatively expressive preference languages.
We also examined different notions of optimality, and proved relationships between them.
Our experimental results confirm that the preference inference/consistency algorithm is fast,
and examined the problem of generating optimal solutions, where it was found that generating the sets of possibly optimal and possibly strictly optimal solutions were significantly faster than generating the undominated solutions.

There are other common forms of preference statement that are strongly compositional, and for which the greedy algorithm will enable checking consistency.
For instance, a restriction on the allowed value orderings of each variable is always strongly compositional.
This can include,
for instance,
highly disjunctive statements, for instance, structural properties of the value orderings, such as being single-peaked \cite{Conitzer09}.
Certain kinds of restrictions on variable orderings
are also strongly compositional.

%\commentph
More formally, regarding the value ordering restrictions,
for each variable $X\in V$
let $\calP_X$ be a set of total orders on $\und{X}$, and let
$\calP$ be $\set{\calP_X \st X\in V}$.
We say that lex model $\pi$, equalling
$(Y_1, \ge_{Y_1}), \ldots, (Y_k, \ge_{Y_k})$,
satisfies $\calP$ if and only if
each $\ge_{Y_i}$ is in $\calP_{Y_i}$.
The definitions immediately imply that
$\calP$ is strongly compositional.

%\commentph
Lemma~\ref{le:compos-conjunction} showed that the property of being strongly compositional is (roughly speaking) preserved under conjunction.
Although this is far from being the case for disjunctions in general, some disjunctive statements are strongly compositional;
this includes the weakly strict statements in $\calLpq$,
and restrictions on value orderings.

A lexicographic order can also be viewed as a static variable and value search ordering, for a constraint satisfaction or optimisation problem.
In a configuration problem, for instance, we may want to generate solutions in an ordering which fits in with what one knows about user preferences and the structure of the problem \cite{JunkerM03}.
We can thus take, as input, past user preferences between solutions, or even partial tuples;
or that all solutions of one constraint are preferred to all those satisfying another constraint.
There may be natural restrictions on the set of value orderings,
%for example, with a numerical domain one will sometimes want to assume at least that the value ordering is single peaked or has a single trough.
e.g., that there is a single peak or trough with a numerical domain.
Given such inputs, which are strongly compositional, one can apply our approach to see if there is a compatible lexicographic order, which can then be used to generate solutions in a best first order.

Our approach in Section~\ref{sec:lex-inf-strong-compos} was rather abstract, in that no explicit structure on the language $\calL$  was assumed.
In fact, a closer examination of the proofs of some of the main results
%in that section
suggests that these will hold more generally than for lexicographic orders:
essentially they just depend on %some
basic properties of composition and extension, but not otherwise on the structure of %the
models.
In particular, it is natural to look at the generalisation to conditional lexicographic orders, which are closely related to conditional preference languages \cite{Wilson09}.
It would also be interesting to explore the potential connection with efficient classes of planning problems.

	\section*{Acknowledgments}
	This publication has emanated from research conducted
	with the financial support of Science Foundation Ireland
	(SFI) under Grant Number SFI/12/RC/2289.

\bibliography{reference}

\begin{thebibliography}{10}

\bibitem{AgrawalW00}
R.~Agrawal and E.~L. Wimmers.
\newblock A framework for expressing and combining preferences.
\newblock In {\em Proc. ACM SIGMOD 2000}, pages 297--306, New York, NY, USA, 2000. ACM.

\bibitem{ArrowR86}
K.~Arrow and H.~Raynaud.
\newblock {\em Social Choice and Multicriterion Decision-Making}, volume~1.
\newblock The MIT Press, 1 edition, 1986.

\bibitem{BoothCLMS10}
R.~Booth, Y.~Chevaleyre, J.~Lang, J.~Mengin, and C.~Sombattheera.
\newblock Learning conditionally lexicographic preference relations.
\newblock In {\em Proc. ECAI 2010}, pages 269--274, 2010.

\bibitem{BoutilierBDHP04-CP-NetsATool}
C.~Boutilier, R.~I. Brafman, C.~Domshlak, H.~Hoos, and D.~Poole.
\newblock {CP}-nets: A tool for reasoning with conditional \textit{ceteris paribus} preference statements.
\newblock {\em Journal of Artificial Intelligence Research}, 21:135--191, 2004.

\bibitem{BrauningH12}
M.~Br\"auning and E.~H\"ullermeier.
\newblock Learning conditional lexicographic preference trees.
\newblock In {\em Preference Learning (PL-12), ECAI 2012 workshop}, pages 11--15, 2012.

\bibitem{ChenP07}
L.~Chen and P.~Pu.
\newblock {\em Preference-based organization interfaces: aiding user critiques in recommender systems}, pages 77--86.
\newblock Springer, 2007.

\bibitem{Conitzer09}
V.~Conitzer.
\newblock Eliciting single-peaked preferences using comparison queries.
\newblock {\em Journal of Artificial Intelligence Research}, 35:161--191, 2009.

\bibitem{DombiIV07}
J.~Dombi, C.~Imreh, and N.~Vincze.
\newblock Learning lexicographic orders.
\newblock {\em European Journal of Operational Research}, 183(2):748--756, 2007.

\bibitem{GelainPRVW10}
M.~Gelain, M.~S. Pini, F.~Rossi, K.~B. Venable, and N.~Wilson.
\newblock Interval-valued soft constraint problems.
\newblock {\em Ann. Math. Artif. Intell.}, 58(3-4):261--298, 2010.

\bibitem{GeorgeW16}
A.~George and N.~Wilson.
\newblock Preference inference based on pareto models.
\newblock In {\em Proc. {SUM} 2016}, pages 170--183, 2016.

\bibitem{JunkerM03}
U.~Junker and D.~Mailharro.
\newblock Preference programming: Advanced problem solving for configuration.
\newblock {\em Artificial Intelligence for Engineering Design, Analysis and Manufacturing}, 17(1):13--29, 2003.

\bibitem{KieBling02}
W.~Kie{\ss}ling.
\newblock Foundations of preferences in database systems.
\newblock In {\em Proc. VLDB 2002}, pages 311--322. VLDB Endowment, 2002.

\bibitem{KohliJ07}
R.~Kohli and K.~Jedidi.
\newblock Representation and inference of lexicographic preference models and their variants.
\newblock {\em Marketing Science}, 26(3):380--399, 2007.

\bibitem{Lang02}
J.~Lang.
\newblock From preference representation to combinatorial vote.
\newblock In {\em Proc.~KR 2002}, pages 277--288, 2002.

\bibitem{MontazeryW16}
M.~Montazery and N.~Wilson.
\newblock Learning user preferences in matching for ridesharing.
\newblock In {\em Proc. {ICAART} 2016}, pages 63--73, 2016.

\bibitem{OMahonyThesis13}
C.~O'Mahony.
\newblock Reasoning with sorted-pareto dominance and other qualitative and partially ordered preferences in soft constraints.
\newblock In {\em PhD Thesis, Cork, 2013}. https://cora.ucc.ie/handle/10468/1498, 2013.

\bibitem{SandholmB06}
T.~Sandholm and C.~Boutilier.
\newblock Preference elicitation in combinatorial auctions.
\newblock In {\em Combinatorial Auctions, ed. P. Cramton, Y. Shoham, and R. Steinberg, Chapter 10}, page 233–264. Cambridge, MA: MIT Press, 2006.

\bibitem{TrabelsiWB13}
W.~Trabelsi, N.~Wilson, and D.~Bridge.
\newblock Comparative preferences induction methods for conversational recommenders.
\newblock In {\em Proc. ADT 2013}, pages 363--374, 2013.

\bibitem{Wilson09}
N.~Wilson.
\newblock Efficient inference for expressive comparative preference languages.
\newblock In {\em Proc.~IJCAI-09}, pages 961--966, 2009.

\bibitem{Wilson14}
N.~Wilson.
\newblock Preference inference based on lexicographic models.
\newblock In {\em Proc.~ECAI 2014}, pages 921--926, 2014.

\bibitem{WilsonGOS15}
N.~Wilson, A.-M. George, and B.~O'Sullivan.
\newblock Computation and complexity of preference inference based on hierarchical models.
\newblock In {\em Proc. IJCAI 2015}, pages 3271--3277, 2015.

\bibitem{WilsonOMahonyAICS2011}
N.~Wilson and C.~O'Mahony.
\newblock The relationships between qualitative notions of optimality for decision making under logical uncertainty.
\newblock In {\em Proc. AICS 2011}, pages 66--75, 2011.

\bibitem{WilsonRM15}
N.~Wilson, A.~Razak, and R.~Marinescu.
\newblock Computing possibly optimal solutions for multi-objective constraint optimisation with tradeoffs.
\newblock In {\em Proc. IJCAI 2015}, IJCAI'15, pages 815--821. AAAI Press, 2015.

\bibitem{YamanWLdJ10}
F.~Yaman, T.~J. Walsh, M.~L. Littman, and M.~Desjardins.
\newblock Learning lexicographic preference models.
\newblock {\em Preference learning}, pages 251--272, 2011.

\end{thebibliography}
\end{document}